\documentclass[10pt, two column, twoside]{IEEEtran}
\usepackage[utf8]{inputenc} 
\usepackage[T1]{fontenc}
\usepackage{url}
\usepackage{ifthen}
\usepackage[cmex10]{amsmath}
\usepackage{array} 

\usepackage{enumerate}
\usepackage{amssymb}
\usepackage{amsmath} 
\usepackage[lined,boxed,commentsnumbered, ruled]{algorithm2e}
\usepackage{mathrsfs}
\usepackage{bm}
\usepackage{tikz}
\usepackage{algorithmic}
\usetikzlibrary{arrows}
\usepackage{subfigure}
\usepackage{graphicx,booktabs,multirow}
\usepackage{epstopdf}
\usepackage{subfigure}
\usepackage{multirow}
\usepackage{tabularx} 
\usepackage{booktabs}
\DeclareMathOperator*{\argmin}{arg\,min}

\definecolor{colorhkust}{RGB}{20,43,140}
\definecolor{colortsinghua}{RGB}{116,52,129}
\definecolor{color1}{RGB}{128,0,0}

\usepackage{amsthm}

\newtheorem{lemma}{Lemma}
\newtheorem{theorem}{Theorem}

\newtheorem{definition}{Definition}
\newtheorem{remark}{Remark}

\newcommand{ \brac }[1]{\left[ #1 \right]}
\newcommand{\norm}[2]{\left\| #1 \right\|_{#2}}
\newcommand{\normn}[2]{\| #1 \|_{#2}}

\newcommand{\sgn}{\mathrm{sgn}}

\newcommand{\dist}{\mathrm{dist}}

\newcommand{\abs}[1]{\left| #1 \right|}
\newcommand{\absn}[1]{| #1 |}
\newcommand{ \paren }[1]{ \left( #1 \right) }
\newcommand{ \Brac }[1]{\left\lbrace #1 \right\rbrace}
\newcommand{\wt}{\widetilde}

\newcommand{\co}{\mathcal{O}}
\begin{document}

\title{Blind Over-the-Air Computation and Data Fusion
       via Provable Wirtinger Flow}

   \author{Jialin Dong, \textit{Student Member}, \textit{IEEE}, Yuanming~Shi, \textit{Member}, \textit{IEEE},
     and Zhi Ding,  \textit{Fellow}, \textit{IEEE} 
        \thanks{J. Dong and Y. Shi are with the School of Information Science and Technology, ShanghaiTech University, Shanghai 201210, China (e-mail: \{dongjl, shiym\}@shanghaitech.edu.cn). }
\thanks{Z. Ding is with the Department of Electrical and Computer Engineering, University of California at Davis, Davis, CA 95616, USA (e-mail: zding@ucdavis.edu).}
        
    }
        
        \maketitle

\begin{abstract}
Over-the-air computation (AirComp) 
shows great promise to support fast data fusion 
in Internet-of-Things (IoT) networks. AirComp 
typically computes desired functions of distributed sensing data 
by exploiting superposed data transmission in
multiple access channels. To overcome its reliance on
channel station information (CSI), 
this work proposes a novel \emph{blind over-the-air computation} 
(BlairComp) without requiring CSI access, particularly 
for low complexity and low latency IoT networks. 
To solve the resulting non-convex optimization problem
without the initialization dependency exhibited by 
the solutions of 
a number of recently proposed efficient algorithms,
we develop a Wirtinger flow solution to the
BlairComp problem based on {\emph{random initialization}}.
To analyze the resulting efficiency, we prove its
statistical optimality and global convergence guarantee. 
Specifically, in the first stage of the algorithm, 
the iteration of randomly
initialized Wirtinger flow given sufficient data samples
can enter a local region
that enjoys strong convexity and strong
smoothness within a few iterations.  
We also prove the estimation error of BlairComp in the local region to be sufficiently small. 
We show that, at the second stage of the algorithm, its
estimation error decays exponentially at a linear convergence rate.

\end{abstract}
\section{Introduction}
The broad range of Internet-of-Things (IoT) applications continues to 
contribute substantially to 
the economic development and the improvement of 
our lives \cite{al2015internet}. In particular, the wirelessly
networked sensors are growing at an unprecedented rate, making 
data aggregation highly critical for IoT services \cite{andrews2014will}. For large scale wireless networking of
sensor nodes,  orthogonal multiple access protocols are highly
impractical because of their low spectrum utilization efficiency 
for IoT and the excessive network latency \cite{huang2018}. 
In response, the concept of \emph{over-the-air computation} (AirComp)
has recently been considered for computing a class of 
{nomographic functions}, such as arithmetic mean, 
weighted sum, geometric mean and Euclidean norm of distributed 
sensor 
data via concurrent, instead of the 
sequential, node transmissions \cite{mario2015}. 
AirComp exploits the natural superposition of co-channel
transmissions from multiple data source nodes 
\cite{nazer2011compute}.

There have already been a number of published works 
related to AirComp.  Among them, one research thread 
takes on the information theoretic view and focuses on achievable computation rate under structured coding schemes.
Specifically, in the seminal work of \cite{nazercomputation}, 
linear source coding was designed to reliably compute a function of 
distributed sensor data transmitted over the multiple-access 
channels (MACs).  
Lattice codes were adopted in \cite{nazercomputation, soundararajan2012communicating} to compute the sum of source signals over MACs efficiently. 
Leveraging lattice coding, a compute-and-forward relaying scheme \cite{nazer2011compute} 
was proposed for relay assisted networks. 
On the other hand, a different line of studies \cite{xiao2008linear,wang2011distortion} investigates the error of distributed estimation in wireless sensor networks. 
In particular, linear decentralized estimation 
was investigated in \cite{xiao2008linear} for coherent multiple access channels. 
Power control was investigated in \cite{wang2011distortion} to 
optimize the estimation distortion. 
It was shown in \cite{goldenbaum2013harnessing} that pre- and post-processing functions
enable the optimization of computation performance
by harnessing the interference for function computations.
Another more recent line of studies focused on designing transmitter and receiver matrices in order to minimize the distortion error when computing desired functions. 
Among others,  MIMO-AirComp equalization and channel feedback
techniques for spatially multiplexing multi-function computation have been proposed \cite{huang2018}.  Another work developed a novel transmitter design
at the multiple antennas IoT devices with zero-forcing beamforming \cite{chen2018over}.

However, the main limitation of current AirComp is the dependence on channel-state-information (CSI), which leads to high latency and significant overhead in the massive Internet-of-Things networks with a large amount of devices. Even though the work \cite{goldenbaum2014channel} has proposed a type of CSI at sensor nodes, called {No CSI},  the receiver still needs to obtain the statistical channel knowledge. Recently, blind demixing has become a powerful tool to elude channel-state-information (i.e., without channel estimation at both transmitters and receivers) thereby enabling low-latency communications \cite{ling2017regularized,ling2015blind,dong}. Specifically, in blind demixing, a sequence of source signals can be recovered from the sum of bilinear measurements without the knowledge of  channel information \cite{dong2018}. Inspired by the recent progress of blind demixing, in this paper, we shall propose a novel \emph{blind over-the-air computation} (BlairComp)  scheme for low-latency data aggregation, thereby computing the desired function (e.g., arithmetic mean) of sensing data vectors without the prior knowledge of channel information. However, the BlairComp problem turns out to be a highly intractable nonconvex optimization problem due to the bilinear signal model. 

There is a growing body of recent works to tame the nonconvexity in solving the high-dimensional bilinear systems. Specifically, semidefinite programming was developed in \cite{ling2015blind} to solve the blind demixing problem by lifting the bilinear model into the matrix space. However, it is computationally prohibitive for solving large-scale problem due to the high computation and storage cost. To address this issue, the nonconvex algorithm, e.g., regularized gradient descent with spectral initialization \cite{ling2017regularized}, was further developed to optimize the variables in the natural vector space. Nevertheless, the theoretical guarantees for the regularized gradient \cite{ling2017regularized} provide pessimistic convergence rate and require carefully-designed initialization. The Riemannian trust-region optimization algorithm without regularization was further proposed in  \cite{dong} to improve the convergence rate. However,  the second-order algorithm brings unique challenges in providing statistical guarantees. Recently, theoretical guarantees concerning
regularization-free Wirtinger flow with spectral initialization for blind demixing  was provided in \cite{dong2018}. However, this regularization-free method still calls for spectral initialization. To find a natural implementation
for the practitioners that works equally well as
spectral initialization, in this paper, we shall propose to solve the
BlairComp problem via {\it{randomly initialized Wirtinger flow}}
with provable optimality guarantees. 

Based on the random initialization strategy, a line of research studies the benign global landscapes for the high-dimensional nonconvex estimation problems, followed by designing generic saddle-point escaping algorithms, e.g., noisy stochastic gradient descent \cite{ge2015escaping}, trust-region method \cite{sun2016geometric}, perturbed gradient descent \cite{jin2017escape}.
With sufficient samples, these algorithms are guaranteed to converge globally for phase retrieval \cite{sun2016geometric}, matrix recovery \cite{bhojanapalli2016global}, matrix sensing \cite{ge2017no}, robust PCA \cite{ge2017no} and shallow neural networks \cite{soltanolkotabi2018theoretical}, where all local minima are provably as good as global and all the saddle points are strict. However,
the theoretical results developed in \cite{ge2015escaping,sun2016geometric,jin2017escape, bhojanapalli2016global, ge2017no, soltanolkotabi2018theoretical} are fairly general and may yield pessimistic convergence rate guarantees. Moreover, these saddle-point escaping algorithms are more complicated for implementation than the natural vanilla gradient descent or Wirtinger flow. To advance the theoretical analysis for gradient descent with random initialization, the fast global convergence guarantee concerning randomly initialized gradient descent for phase retrieval has been recently provided in \cite{chen2018}.

In this paper, our main contribution is to provide the global convergence guarantee concerning Wirtinger flow with random initialization for solving the nonconvex BlairComp problem. 
It turns out that, for BlairComp, the procedure of Wirtinger flow with random initialization can be separated into two stages:
\begin{itemize}
\item Stage I: the estimation error is nearly stable, which takes only a few iterations,
\item Stage II: 
the estimation error decays exponentially at a linear convergence rate.
\end{itemize}
In addition, we identify the exponential growth of the magnitude ratios of the signals to perpendicular components, which explains why Stage I lasts only for a few iterations. 

\subsubsection*{Notations}Throughout this paper, $f({n}) =  O(g(n))$ or $f(n)\lesssim g(n)$ denotes that there exists a constant $c>0$ such that $|f(n)|\leq c|g(n)|$ whereas $f(n)\gtrsim g(n)$ means that there exists a constant $c>0$ such that $|f(n)|\geq c|g(n)|$. $f(n)\gg  g(n)$ denotes that there exists some sufficiently large constant $c>0$ such that $|f(n)|\geq c|g(n)|$. In addition, the notation $f(n)\asymp g(n)$ means that there exists constants $c_1, c_2>0$ such that $c_1|g(n)|\leq |f(n)|\leq c_2|g(n)|$.  Let superscripts $(\cdot)^{\top}$ and
$(\cdot)^{\mathsf{H}}$ denote the transpose and conjugate transpose of a matrix/vector,  
respectively. Let the superscript $(\cdot)^*$ denote the conjugate transpose of a complex number.

\section{Problem Formulation}\label{form}
Blind over-the-air computation (BlairComp) aims to facilitate
low-latency data aggregation in IoT networks without a priori knowledge of CSI. 
This is achieved by computing the desired functions of the distributed sensing data 
based on the natural signal superposition of transmission over 
multi-access channels.

\subsection{Blind Over-the-Air Computation}
We consider a wireless sensor network consisting of $s$ active sensor nodes and a single fusion center. 
Let $\bm{d}_i=[d_{i1}\;\cdots,\; d_{i,N}]^\top \in\mathbb{C}^N$ 
denote the sensor data vector 
collected at the $i$-th node. 
The fusion center, through AirComp, aims to compute nomographic functions of 
distributed data that can be decomposed as
\cite{goldenbaum2013harnessing}
\setlength\arraycolsep{2pt} 
\begin{align}\label{eq:nom_fun}
\mathcal{H}_\ell(d_{1\ell},\cdots,d_{s\ell}) = \mathcal{F}_\ell\big(\sum\nolimits_{i=1}^s\mathcal{G}_{i\ell}(d_{i\ell})\big),\;\ell =1,\cdots,N.
\end{align}
Function $\mathcal{G}_{i\ell}(\cdot):\mathbb{C}\rightarrow \mathbb{C}$ denotes 
the pre-processing function by the sensor nodes and
 $\mathcal{F}_{\ell}(\cdot):\mathbb{C}\rightarrow \mathbb{C}$ denotes 
the post-processing function at the fusion center. Typical 
nomographic functions by AirComp include the arithmetic mean, 
weighted sum, geometric mean, polynomial, Euclidean norm \cite{goldenbaum2013harnessing}.

In this work, we focus on a specific nomographic function
\begin{equation}
\bar{\bm{\theta}}= \sum_{i=1}^s\bar{\bm{x}}_i,
\end{equation} 
where  $\bar{\bm{x}}_i= [\mathcal{G}_{i1}(d_{i1}),\cdots,\mathcal{G}_{iN}(d_{iN})]^{\top}\in\mathbb{
C}^N$ is the preprocessed data vector transmitted by the $i$-th node. 
Over $m$ channel access opportunities (e.g., time slots), 
the received signals at fusion center in the frequency domain 
can be written as \cite{ling2017regularized,dong}
\begin{align}\label{eq:bilinear}
y_j  = \sum_{i=1}^{ s}\bm{b}_j^{\mathsf{H}}
\bar{\bm{h}}_i\bar{\bm{x}}_i^{\mathsf{H}}\bm{a}_{ij}+e_j,~ 1\leq j \leq m,
\end{align}
where 
 $\bm{b}_j\in\mathbb{C}^K, \; j=1,\;\cdots,\; m$ are the access vectors,
 $\bm{h}_i\in\mathbb{C}^K$ the CSI vectors, 
and ${e}_j$ is an independent circularly symmetric complex
Gaussian measurement noise.
%

To compute the desired functions via BlairComp without knowledge
of $\{\bar{\bm{h}}_i\}$, we can consider a precoding scheme 
with randomly selected known vectors $\bm{a}_{ij}\in\mathbb{C}^N$ follows i.i.d. circularly symmetric complex normal distribution 
$ \mathcal{N}(\bm{0},0.5\bm{I}_N)+i\mathcal{N}(\bm{0},
0.5\bm{I}_N)$ for $1\leq
i\leq s,1\leq j\leq m$. 
Furthermore, the first $K$ columns of the unitary discrete Fourier transform (DFT) matrix $\bm{F}$ form the known matrix $\bm{B}:=[\bm{b}_1,\cdots,\bm{b}_m]^{\mathsf{H}}\in \mathbb{C}^{m\times K}$  \cite{ling2017regularized}. The target of BlairComp is to compute the desired function vector $\bar{\bm{\theta}}$
via concurrent transmissions without channel information, thereby providing
low-latency data aggregation in the IoT networks.

\subsection{Multi-Dimensional Nonconvex Estimation}
One way to estimate the result vector $\bar{\bm{\theta}}$
from the received signals $y_i$ in (\ref{eq:bilinear}) is to
use  $\bm{\theta}=\sum_{i=1}^{s}\omega_i\bm{x}_i$ with $\omega_i\in\mathbb{C}$ as the ambiguity alignment parameter,
for which
one could first solve the bilinear optimization problem:
\begin{eqnarray}
\label{lea_squ_sdp}
\mathscr{P}: \mathop{\textrm{minimize }}_{\{\bm{h}_i\},\{\bm{x}_i\}}f(\bm{h},\bm{x}):=\sum_{j=1}^m\Big|\sum_{i=1}^{ s}\bm{b}_j^{\mathsf{H}}\bm{h}_i\bm{x}_i^{\mathsf{H}}\bm{a}_{ij}-{y}_j\Big|^2,
\end{eqnarray}
which is highly nonconvex. 
To measure the computation accuracy for BlairComp problem, 
we define the following metric 
\begin{align}\label{eq:relativeerror}
\mathrm{error}(\bm{\theta},\bar{\bm{\theta}}) = \frac{\norm{\sum_{i =1}^s\omega_i\bm{x}_i - \sum_{i =1}^s\bar{\bm{x}}_i}{2}}{\norm{ \sum_{i =1}^s\bar{\bm{x}}_i}{2}}.
\end{align}
Note that $\{\omega_i,\, i=1, \cdots, s\}$ are ambiguity alignment parameters 
such that \begin{align}\label{align}
\omega_i = \argmin\limits_{\omega_i\in\mathbb{ C}}\left({{\|(\omega_i^*)^{-1}\bm{h}_i-\bar{\bm{h}}_i \|_2^2+\|\omega_i\bm{x}_i - \bar{\bm{x}}_i\|_2^2 }}\right). 
\end{align}
To estimate $\omega_i$, one reference symbol in $\bm{x}_i$ is
needed.

In this paper, we shall propose to solve the high-dimensional BlairComp problem $\mathscr{P}$ via Wirtinger flow with 
{\it{random initialization}}. Our main contribution is to provide the 
statistical optimality and convergence guarantee for the randomly 
initialized Wirtinger flow algorithm by
exploiting the benign geometry of the high-dimensional BlairComp 
problem.

\section{Main Approach}
In this section, we first propose an algorithm based on
randomly initialized Wirtinger flow 
to solve the BlairComp problem $\mathscr{P}$. We shall present
a statistical analysis to demonstrate the optimality of this algorithm 
for solving the high-dimensional nonconvex estimation problem.
 
\subsection{Randomly Initialized Wirtinger Flow Algorithm }
Wirtinger flow with random initialization is an iterative algorithm 
with a simple gradient descent update procedure without 
regularization. Specifically, the gradient step  of  Wirtinger 
flow is represented by the notion of Wirtinger derivatives \cite{candes2015phase}, i.e., the derivatives of real valued 
functions over complex variables.

To simplify the notations, we denote $f(\bm{z}) := f(\bm{h},\bm{x})$, 
where
 \begin{align}\label{eq:z}
 \bm{z} = \left[\begin{matrix}
 \bm{z}_1\\\cdots\\\bm{z}_s
 \end{matrix}
 \right]\in \mathbb{ C}^{s(N+K)}~{\textrm{with}}~\bm{z}_i = \left[\begin{matrix}
 \bm{h}_i\\\bm{x}_i
 \end{matrix}\right]\in \mathbb{ C}^{N+K}.
 \end{align}For each $i = 1,\cdots,s$, $\nabla_{\bm{h}_i}f(\bm{z})$ and $\nabla_{\bm{x}_i}f(\bm{z})$ denote the Wirtinger gradient of $f(\bm{z})$ with respect to $\bm{h}_i$ and $\bm{x}_i$ respectively as:
 \begin{subequations}\label{gradient}
        \begin{align}
        \nabla_{\bm{h}_i}f(\bm{z}) &= \sum_{j=1}^m\bigg(\sum_{k=1}^{ s}\bm{b}_j^{\mathsf{H}}\bm{h}_k\bm{x}_k^{\mathsf{H}}\bm{a}_{kj}-{y}_j\bigg)\bm{b}_j\bm{a}_{ij}^{\mathsf{H}}\bm{x}_i,\label{g1}\\
        \nabla_{\bm{x}_i}f(\bm{z}) &= \sum_{j=1}^m{\bigg(\sum_{k=1}^{ s}\bm{h}_k^{\mathsf{H}}\bm{b}_j\bm{a}_{kj}^{\mathsf{H}}\bm{x}_k-{y}_j^*\bigg)}\bm{a}_{ij}\bm{b}_j^{\mathsf{H}}\bm{h}_i.\label{g2}
        \end{align}
 \end{subequations}
 In light of the Wirtinger gradient (\ref{gradient}), the update rule of Wirtinger flow uses a stepsize $\eta>0$ via
 \begin{align}\label{eq:update role}
 \left[\bm{h}_i^{t+1}\atop\bm{x}_i^{t+1} \right]
        = \left[\bm{h}_i^{t}\atop\bm{x}_i^{t}\right]-\eta \left[\frac{1}{\|\bm{x}^t_i\|_2^2}\nabla_{\bm{h}_i}f(\bm{z}^t)\atop\frac{1}{\|\bm{h}_i^t\|_2^2}\nabla_{\bm{x}_i}f(\bm{z}^t) \right],\, i = 1,\cdots,s.
 \end{align}

Before proceed to theoretical analysis, we first present an example to
illustrate the practical efficiency of Wirtinger flow with random initialization for solving problem $\mathscr{P}$ (\ref{lea_squ_sdp}). The ground truth 
values $\{\bar{\bm{h}}_i, \bar{\bm{x}}_i\}$ 
and initial points $\{\bm{h}_i^{0}, \bm{x}_i^{0}\}$
are randomly generated according to
\begin{align}
\bar{\bm{h}}_i\sim\mathcal{N}(\bm{0},K^{-1}\bm{I}_K),~\bar{\bm{x}}_i\sim\mathcal{N}(\bm{0},N^{-1}\bm{I}_N),\\
\bm{h}_i^{0}\sim\mathcal{N}(\bm{0},K^{-1}\bm{I}_K),~\bm{x}_i^0\sim\mathcal{N}(\bm{0},N^{-1}\bm{I}_N),\label{eq:initial}
\end{align}
for $i = 1,\cdots, s$.  In all our simulations, we set $K = N$ and 
normalize $\normn{\bar{\bm{h}}_i}{2} = \normn{\bar{\bm{x}}_i}{2} = 1$ 
for $i = 1,\cdots,s$.  Specifically, for each value of 
$K\in\{20,80,160,200\}$, $s = 10$ and $m = 50K$, 
the design vectors $\bm{a}_{ij}$'s and $\bm{b}_j$'s for 
each $1\leq i\leq s, 1\leq j\leq m$, are generated according to the descriptions in Section \ref{form}. With the chosen step size  
$\eta = 0.1$ in all settings, Fig. \ref{fig1} shows the 
relative error, i.e.,
$
\mathrm{error}(\bm{\theta}^t,\bar{\bm{\theta}})
$ (\ref{eq:relativeerror}), versus the iteration count.
We observe the convergence of Wirtinger flow with random initialization 
exhibits two stages:  Stage I:  within dozens of iterations, 
the relative error remains nearly flat, Stage II:  
the relative error shows exponential decay
despite the different problem sizes. 

\begin{figure}[htb]
        \subfigure[]{\includegraphics[width=1.72in]{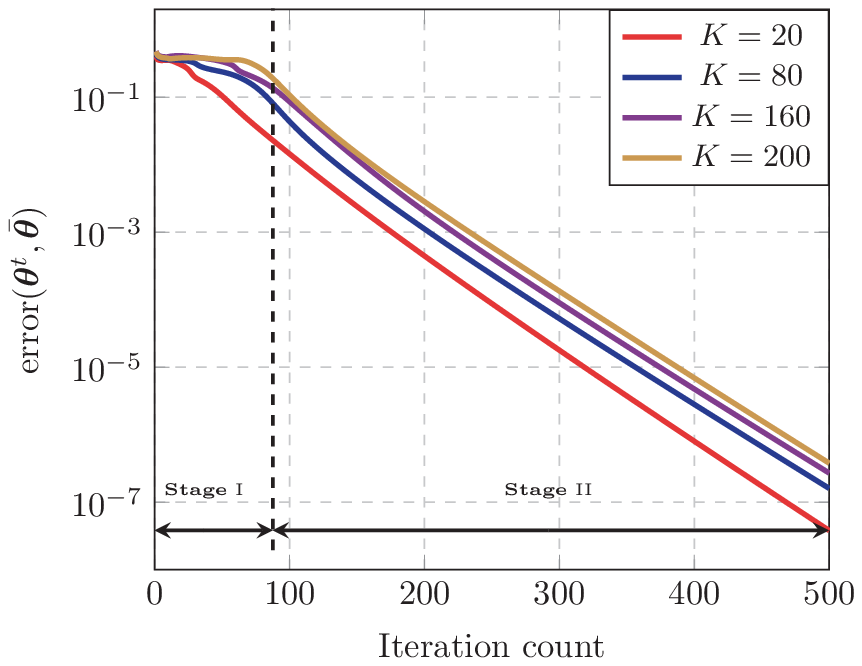}
                \label{fig1}}
        \hspace*{-0.4cm}
        \subfigure[]{\includegraphics[width=1.72in]{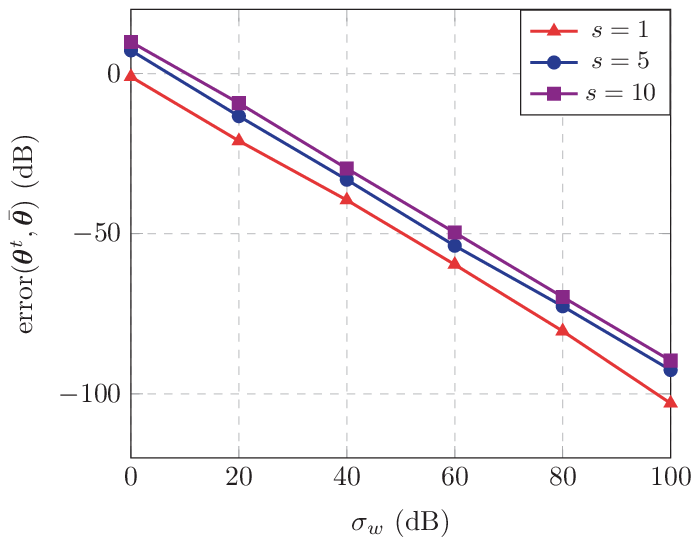}
                \label{fig4}}
        \caption{(a) Linear convergence rate of randomly initialized Wirtinger flow, (b) relative error versus $\sigma_{w}$ (dB). }  
\end{figure}   

In practical scenario, the estimation error of ambiguity alignment parameters would have influences on the relative error, i.e.,
$
\mathrm{error}(\bm{\theta}^t,\bar{\bm{\theta}})
$ (\ref{eq:relativeerror}). Hence, we illustrate the relationship between the estimation error of ambiguity alignment parameters and the relative error via the following experiment. Let $K=10$, $m=100$, the step size be  
$\eta = 0.1$ and the number of users $s\in\{1, 5, 10\}$. 
In each iteration, for $i = 1,\cdots, s$, the estimated ambiguity alignment parameter $\hat{w}_i$ is given by
$
\hat{w}_i = w_i + e_{w_i},
$
where $w_i$ is given by (\ref{align}) and $e_{w_j}\sim \mathcal{N}({0},0.5\sigma_{w}^{-1})+i\mathcal{N}({0},
0.5\sigma_{w}^{-1})$ is the additive noise. In the experiment, the parameter $\sigma_{w}$ varies from $1$ to $10^{5}$. Fig. \ref{fig4} shows the relative error $
\mathrm{error}(\bm{\theta}^t,\bar{\bm{\theta}})
$ versus the parameter $\sigma_{w}$. Both the relative error and the parameter $\sigma_{w}$ are shown in the dB scale.
As we can see, the relative error scales linearly with the parameter $\sigma_{w}$.

  \subsection{Theoretical Analysis}
  To present the main theorem, we first introduce several fundamental definitions. Specifically,  the incoherence parameter \cite{ling2017regularized}, which characterizes the incoherence between $\bm{b}_j$ and $\bm{h}_i$ for $1\leq i \leq s, 1\leq j\leq m$.
  \begin{definition}[Incoherence for BlairComp]\label{inco} Let the incoherence parameter $\mu$ be the smallest number such that
 $
        \max_{1\leq i \leq s, 1\leq j\leq m}\frac{|\bm{b}^{\mathsf{H}}_j\bar{\bm{h}}_i|}{\|\bar{\bm{h}}_i\|_2}\leq \frac{\mu}{\sqrt{m}}.
$
  \end{definition}
Let $\widetilde{\bm{h}}^t_i$ and $\widetilde{\bm{x}} ^t_i
  $, respectively,  denote
  \begin{align}
\widetilde{\bm{h}}^t_i = ({{\omega_i^t}^*})^{-1}\bm{h}_i^t\quad\text{and}\quad\widetilde{\bm{x}} ^t_i= \omega_i^t\bm{x}_i^t,\; i = 1,\cdots, s, 
\label{eq:def-xl}
  \end{align}
where $\omega_i^t$'s are alignment parameters. We further define the norm of the signal component and the perpendicular component with respect to $\bm{h}_i^t$ for $i = 1,\cdots, s,$ as
   \begin{align}
   {\alpha}_{\bm{h}_i^t}
   &:=\langle\bar{\bm{h}}_i,\widetilde{\bm{h}}^t_i\rangle/\|\bar{\bm{h}}_i\|_2,\label{eq:alpha_h}\\
   {\beta}_{\bm{h}_i^t}&:=\left\Vert \widetilde{\bm{h}}^t_i-\frac{\langle\bar{\bm{h}}_i,\widetilde{\bm{h}}^t_i\rangle}{\|\bar{\bm{h}}_i\|_2^2}\bar{\bm{h}}_i\right\Vert _{2},\label{eq:beta_h}
   \end{align} respectively.
  Here, $\omega_i$'s are the alignment parameters.
  Similarly, the norms of the signal component and the perpendicular component with respect to $\bm{x}_i^t$ for $i = 1,\cdots, s,$ can be represented as

 \begin{align}
        {\alpha}_{\bm{x}_i^t}&:=\langle\bar{\bm{x}}_i,\widetilde{\bm{x}} ^t_i\rangle/\|\bar{\bm{x}}_i\|_2,\label{eq:alpha_x}\\
        {\beta}_{\bm{x}_i^t}&:=\left\Vert \widetilde{\bm{x}} ^t_i-\frac{\langle\bar{\bm{x}}_i,\widetilde{\bm{x}} ^t_i\rangle}{\|\bar{\bm{x}}_i\|_2^2}\bar{\bm{x}}_i\right\Vert _{2},\label{eq:beta_x}
        \end{align}
  respectively.
 
  Without loss of generality, we assume $\|\bar{\bm{h}}_i\|_2 = \|\bar{\bm{x}}_i\|_2 = q_i$ ($0<q_i\leq 1$) for $i=  1,\cdots ,s$ and ${\alpha}_{\bm{h}_i^0}$, ${\alpha}_{\bm{x}_i^0}>0$ for $i=  1,\cdots ,s$. Define the condition number $\kappa:=\frac{\max_i \|\bar{\bm{x}}_i\|_2 }{\min_i \|\bar{\bm{x}}_i\|_2 }\geq 1$ with $\max_i\|\bar{\bm{x}}_i\|_2=1$. Then the main theorem is presented in the following. 
   
 \begin{theorem}\label{main_thm}Assume that the initial points obey (\ref{eq:initial}) for $i = 1,\cdots,s$ and the stepsize $\eta>0$ satisfies $\eta \asymp s^{-1}$.  Suppose that the sample size satisfies $m\geq C \mu^2s^2\kappa^4\max\{K,N\}\log^{12}m$ for some sufficiently large constant $C>0$. Then with probability at least $1-c_1m^{-\nu}-c_1me^{-c_2N}$ for some constants $\nu, c_1, c_2>0$,  there exists a sufficiently small constant $0\leq \gamma\leq 1$ and $T_\gamma\lesssim s\log (\max{\{K,N\}})$ such that
        \begin{enumerate}
                \item The randomly initialized Wirtinger
flow makes the estimation error decays exponentially, i.e.,
               \begin{eqnarray}
 \mathrm{error}(\bm{\theta}^{t},\bar{\bm{\theta}})\leq \gamma\left(1-\frac{\eta}{16\kappa}\right)^{t-T_\gamma},~ t\geq T_\gamma,
          \end{eqnarray}
                \item The magnitude ratios of the signal component
                to the perpendicular component with respect to $\bm{h}_i^t$ and $\bm{x}_i^t$
                obey
                \begin{subequations}\label{eq:exp_ratio}
                \begin{align}
        \max_{1\leq i\leq s}    \frac{ {\alpha}_{\bm{h}_i^t}}{ {\beta}_{\bm{h}_i^t}}\gtrsim \frac{1}{\sqrt{K\log K}}(1+c_3\eta)^t,\\
                \max_{1\leq i\leq s}    \frac{ {\alpha}_{\bm{x}_i^t}}{ {\beta}_{\bm{x}_i^t}}\gtrsim \frac{1}{\sqrt{N\log N}}(1+c_4\eta)^t,
                \end{align}     \end{subequations}
                respectively, where $t=0,1,\cdots$ for some constants $c_{3},c_4>0$.
                \item The normalized root mean square error ${\rm{RMSE}}(\bm{x}_i^{t},\bar{\bm{x}}_i) = {\beta_{\bm{x}_i^t}}/{\normn{\bm{x}_i^t}{2}}$ for $i = 1,\cdots, s$ obeys 
                \begin{align}\label{rmse}
                {\rm{RMSE}}(\bm{x}_i^{t},\bar{\bm{x}}_i)  \lesssim{\sqrt{N\log N}}(1+c_4\eta)^{-t},
                \end{align}
                for some constant $c_4>0$.
        \end{enumerate}
 \end{theorem}
   Theorem \ref{main_thm} provides precise statistical analysis on the computational efficiency of Wirtinger
flow with random initialization. Specifically, in Stage I, it takes $T_\gamma = \co( s\log (\max{\{K,N\}}))$ iterations for randomly initialized Wirtinger
flow to reach sufficient small relative error, i.e., $ \mathrm{error}(\bm{\theta}^{T_{\gamma}},\bar{\bm{\theta}})\leq \gamma$ where $\gamma>0$ is some sufficiently small constant.   The short duration of Stage I is own to the exponential growth of the magnitude ratio of the signal component to the perpendicular components.  Moreover, in Stage II, it takes $\co(s\log (1/\varepsilon))$ iterations to reach $\varepsilon$-accurate solution at a linear convergence rate. Thus, the iteration complexity of randomly initialized Wirtinger
flow is guaranteed to be $\co(s\log(\max{\{K,N\}})+s\log (1/\varepsilon))$ as long as the sample size exceeds $m\gtrsim s^2\max{\{K,N\}}\mathrm{poly}\log(m)$.
   
   To further illustrate the relationship between the signal component $\alpha_{\bm{h}_i}$ (resp. $\alpha_{\bm{x}_i}$) and the perpendicular component $\beta_{\bm{h}_i}$ (resp. $\beta_{\bm{x}_i}$) for $i = 1,\cdots,s$,  we provide the simulation results under the setting of $K = N = 10$, $m = 50K$, $s=4$ and $\eta = 0.1$ with  $\|\bar{\bm{h}}_i\|_2 =\|\bar{\bm{x}}_i\|_2 =1$ for $1\leq i\leq s$. In particular, $\alpha_{\bm{h}_i}$,  $\beta_{\bm{h}_i}$ versus iteration count (resp. $\alpha_{\bm{h}_i}$,  $\beta_{\bm{h}_i}$ versus iteration count) for $i = 1,\cdots,s$ is demonstrated in Fig. \ref{fig2h} (resp. Fig. \ref{fig2x}). Consider Fig. \ref{fig1}, Fig. \ref{fig2h} and Fig. \ref{fig2x} collectively, it shows that despite the rare decline of the estimation error, i.e., $ \mathrm{error}(\bm{\theta}^t,\bar{\bm{\theta}})$, during Stage I, the size of the signal component, i.e., $\alpha_{\bm{h}_i}$ and $\alpha_{\bm{x}_i}$ for each $i = 1,\cdots,s$, exponentially increase and the signal component becomes dominant component at the end of Stage I. Furthermore, the exponential growth of the ratio $      { {\alpha}_{\bm{h}_i}}/{ {\beta}_{\bm{h}_i}}$ (resp. ${ {\alpha}_{\bm{x}_i}}/{ {\beta}_{\bm{x}_i}}$) for each $i = 1,\cdots, s$ is illustrated in Fig. \ref{fig3h} (resp. Fig. \ref{fig3x}).
\begin{figure}[htbp]
        \subfigure[]{\includegraphics[width=1.62in]{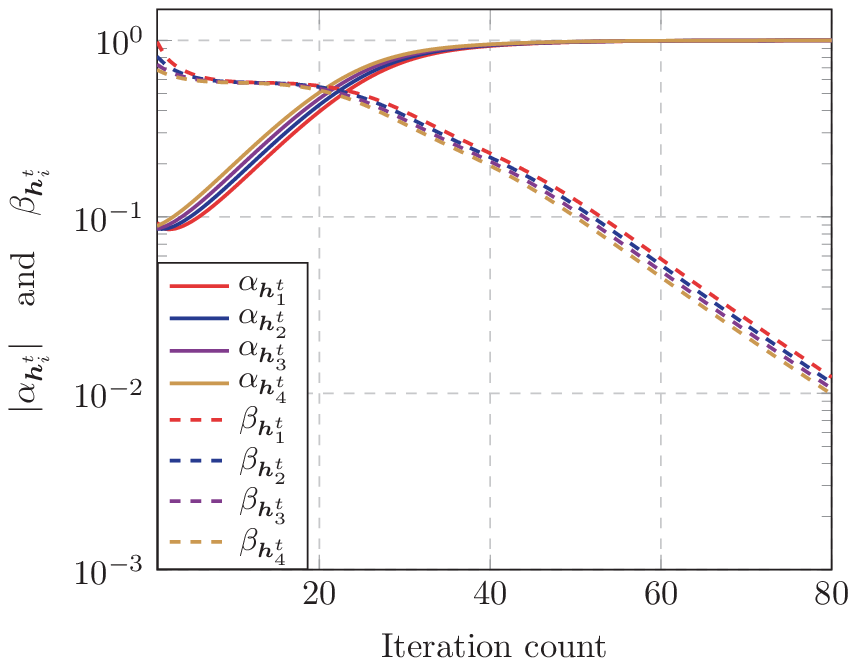}
                \label{fig2h}}\hspace*{-4mm}
        \subfigure[]{\includegraphics[width=1.62in]{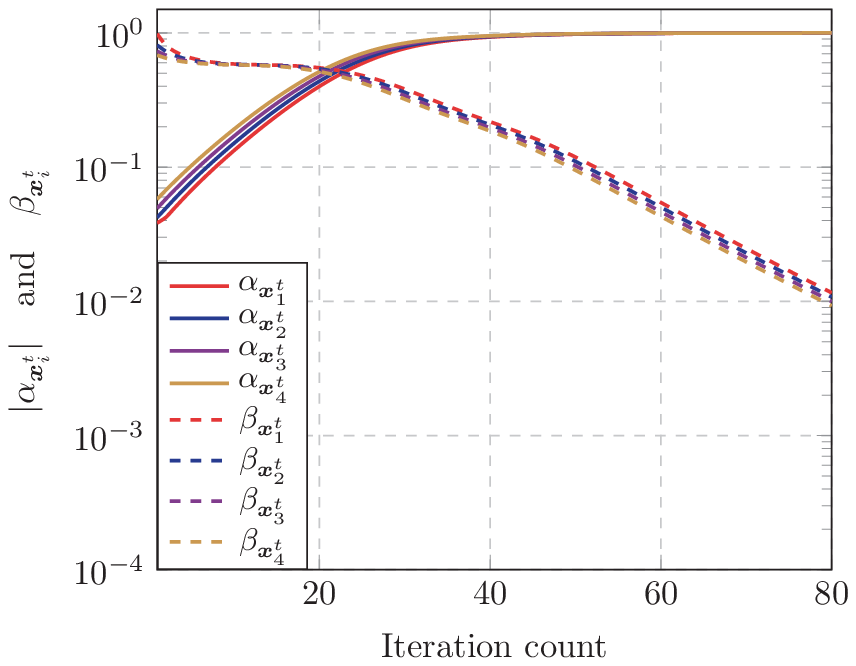}
                \label{fig2x}}\\
        \subfigure[]{\includegraphics[width=1.62in]{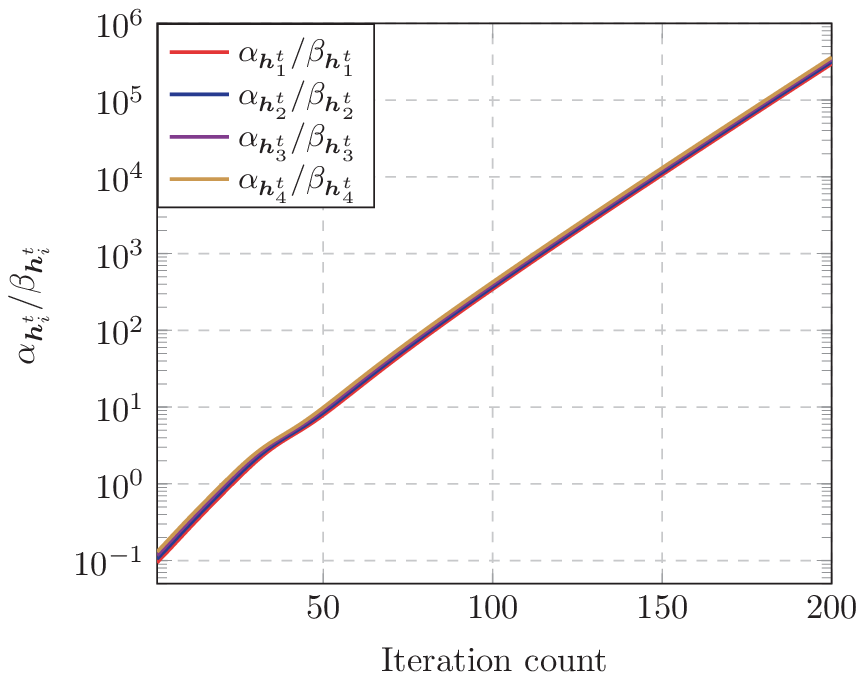}
                \label{fig3h}}\hspace*{-4mm}
        \subfigure[]{\includegraphics[width=1.62in]{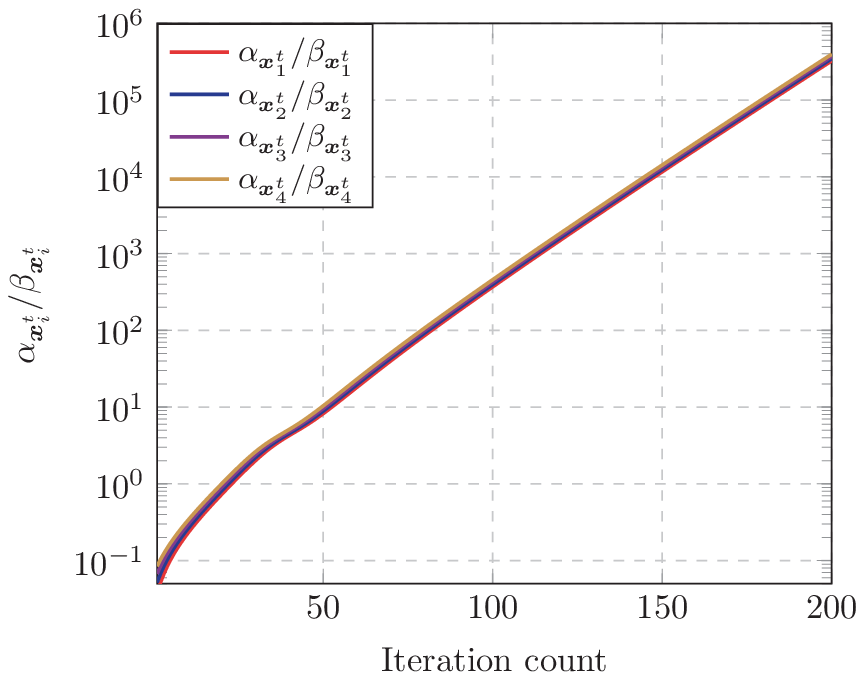}
                \label{fig3x}}
        \caption{Numerical example of signal versus perpendicular components.}
\end{figure} 

   \section{Dynamics Analysis}
In this section, we prove the main theorem by investigating the dynamics of the iterates of Wirtinger
flow with random initialization.  The steps of proving Theorem \ref{main_thm} are summarized as follows.
\begin{enumerate}
 \item \textbf{Stage I}:  
\begin{itemize}
        \item \textbf{Dynamics of population-level state evolution.} Provide the population-level state evolution of $\alpha_{\bm{x}_i}$ (\ref{eq:state-evolution-population-alphax}) and $\beta_{\bm{x}_i}$ (\ref{eq:state-evolution-population-betax}), $\alpha_{\bm{h}_i}$ (\ref{eq:state-evolution-population-alphah}), $\beta_{\bm{h}_i}$ (\ref{eq:state-evolution-population-betah}) respectively, where the sample size approaches infinity. We then develop the approximate state evolution  (\ref{eq:app_state}), which are remarkably close to the population-level state evolution, in the finite-sample regime. See details in Section \ref{section:dym}.
\item \textbf{Dynamics of approximate state evolution.} Show that there exists some $T_\gamma = \co( s\log (\max{\{K,N\}}))$ such that $ \mathrm{error}(\bm{x}^{T_\gamma},\bar{\bm{x}})\leq\gamma$ , if $\alpha_{\bm{h}_i}$ (\ref{eq:alpha_h}), $\beta_{\bm{h}_i}$ (\ref{eq:beta_h}), $\alpha_{\bm{x}_i}$ (\ref{eq:alpha_x}) and $\beta_{\bm{x}_i}$ (\ref{eq:beta_x}) satisfy the approximate state evolution (\ref{eq:app_state}). The exponential growth of the ratio ${ {\alpha}_{\bm{h}_i}}/{ {\beta}_{\bm{h}_i}}$ and ${ {\alpha}_{\bm{x}_i}}/{ {\beta}_{\bm{x}_i}}$ are further demonstrated under the same assumption. Please refer to Lemma \ref{lemma:iterative}.
\item \textbf{Leave-one-out arguments.} Prove that with high probability  $\alpha_{\bm{h}_i}$, $\beta_{\bm{h}_i}$, $\alpha_{\bm{x}_i}$ and $\beta_{\bm{x}_i}$ satisfy the approximate state evolution (\ref{eq:app_state}) if the iterates $\{\bm{z}_i\}$  are independent with $\{\bm{a}_{ij} \}$.  Please refer to Lemma \ref{lemma:state_evo}. To achieve this, the ``near-independence" between $\{\bm{z}_i\}$ and $\{\bm{a}_{ij} \}$ is established via exploiting leave-one-out arguments and some variants of the arguments. Specifically, the leave-one-out sequences and random-sign sequences are constructed in Section \ref{key}. The concentrations between the original and these auxiliary sequences are then provided in Lemma \ref{lemma:xt-xt-l}-Lemma \ref{lemma:xt-xt-l_2}.
\end{itemize}
\item \textbf{Stage II}: \textbf{Local geometry in the region of incoherence and contraction.}
We invoke the prior theory provided in \cite{dong2018} to show local convergence of the random initialized Wirtinger
flow in Stage II.  
\end{enumerate}

Claims (\ref{eq:exp_ratio}) and (\ref{rmse}) are further proven in Section \ref{sec:stage}.

 \subsection{Dynamics of Population-level State Evolution}\label{section:dym}
 In this subsection, we investigate the dynamics of population-level (where we have infinite samples) state evolution of $\alpha_{\bm{h}_i}$ (\ref{eq:alpha_h}), $\beta_{\bm{h}_i}$ (\ref{eq:beta_h}), $\alpha_{\bm{x}_i}$ (\ref{eq:alpha_x}) and $\beta_{\bm{x}_i}$ (\ref{eq:beta_x}). 

Without loss the generality, we assume that $\bar{\bm{x}}_i  = q_i\bm{e}_1$ for $i = 1,\cdots, s$, where $0<q_i\leq 1,i=1,\cdots,s$ are some constants and $\kappa= \frac{\max_i q_i}{\min_i q_i}$, and $\bm{e}_1$ denotes the first standard basis vector. This assumption is based on the rotational invariance of Gaussian distributions.  Since the deterministic nature of $\{\bm{b}_j\}$, the ground truth signals $\{\bar{\bm{h}}_i\}$ (channel vectors) cannot be transferred to a simple form, which yields more tedious analysis procedure. For simplification, for $i = 1,\cdots,s$, we denote \begin{equation}
{x}_{i1}^{t}\qquad\text{and}\qquad\bm{x}_{i\perp}^{t}:=[{x}_{ij}^{t}]_{2\leq j\leq N}\label{eq:defn-xperp-xpara}
\end{equation}
as the first entry and the second through the $N$-th entries of $\bm{x}_i^{t}$,
respectively. Based on the assumption that $\bar{\bm{x}}_i  = q_i\bm{e}_1$ for $i = 1,\cdots, s$, (\ref{eq:alpha_x}) and (\ref{eq:beta_x}) can be reformulated as
\begin{align}
        {\alpha}_{\bm{x}_i}:=\wt{x}_{i1}^{t}\quad\text{and}\quad
        {\beta}_{\bm{x}_i}:=\left\Vert \wt{\bm{x}}_{i\perp}^{t}\right\Vert _{2}.\label{eq:alpha_beta_x1}
\end{align}

To study the population-level state evolution, we start with consider the case where the sequences $\{\bm{z}_i^t\}$ (refer to (\ref{eq:z})) are established via the population gradient, i.e., for $i = 1,\cdots,s$,
\begin{align}
 \left[\bm{h}_i^{t+1}\atop\bm{x}_i^{t+1} \right]
 = \left[\bm{h}_i^{t}\atop\bm{x}_i^{t}\right]-\eta \left[\frac{1}{\|\bm{x}^t_i\|_2^2}\nabla_{\bm{h}_i}F(\bm{z}^t)\atop\frac{1}{\|\bm{h}_i^t\|_2^2}\nabla_{\bm{x}_i}F(\bm{z}^t) \right] ,
\end{align}
where
\[
\nabla_{\bm{h}_i} F(\bm{z}):=\mathbb{E}[\nabla_{\bm{h}_i} f(\bm{h},\bm{x})]=\norm{\bm{x}_i}{2}^2\bm{h}_i-(\bar{\bm{x}}_i^{\mathsf{H}}\bm{x}_i)\bar{\bm{h}}_i,
\]
\[
\nabla_{\bm{x}_i} F(\bm{z}):=\mathbb{E}[\nabla_{\bm{x}_i} f(\bm{h},\bm{x})]=\norm{\bm{h}_i}{2}^2\bm{x}_i-(\bar{\bm{h}}_i^{\mathsf{H}}\bm{h}_i)\bar{\bm{x}}_i.
\]
Here, the population gradients are computed based on the assumption that $\{\bm{x}_i\}$ (resp. $\{\bm{h}_i\}$) and $\{\bm{a}_{ij} \}$ (resp. $\{\bm{b}_j\}$)  are independent with each other.
With simple calculations, the dynamics for both the signal and the perpendicular components with respect to $\bm{x}_i^t$, $i = 1,\cdots,s$ are given as
\begin{subequations}\label{subeq:population-dynamics}
        \begin{align}
        \wt{x}_{i1}^{t+1} & =\left(1-\eta\right) \wt{x}_{i1}^{t}+\eta\frac{q_i^2}{\|\wt{\bm{h}}_i^{t}\|_{2}^{2}}
        \bar{\bm{h}}_i^{\mathsf{H}}\wt{\bm{h}}_i^t,\label{eq:population-alpha-t}\\
        \wt{\bm{x}}_{i\perp}^{t+1} & =\left(1-\eta\right) \wt{\bm{x}}_{i\perp}^{t}.\label{eq:population-beta-t}
        \end{align}
\end{subequations}
Assuming that $\eta>0$ is sufficiently small and $\|\bar{\bm{h}}_i\|_2 = \|\bar{\bm{x}}_i\|_2 = q_i$ ($0< q_i\leq 1$) for $i=  1,\cdots ,s$ and
recognizing that $\|\wt{\bm{h}}_i^{t}\|_{2}^{2}={\alpha}_{\bm{h}_i^{t}}^2+{\beta}_{\bm{h}_i^{t}}^2$,
we arrive at the following population-level state evolution for both
${\alpha}_{\bm{x}_i^{t}}$ and ${\beta}_{\bm{x}_i^{t}}$:
\begin{subequations}\label{subeq:population-iterative_x}
        \begin{align}
\alpha_{\bm{x}_i^{t+1}} & =(1-\eta) \alpha_{\bm{x}_i^{t}}+\eta\frac{{q_i\alpha}_{\bm{h}_i^{t}}}{{\alpha}_{\bm{h}_i^{t}}^2+{\beta}_{\bm{h}_i^{t}}^2},\label{eq:state-evolution-population-alphax}\\
        \beta_{\bm{x}_i^{t+1}} & =(1-\eta )\beta_{\bm{x}_i^{t}}.\label{eq:state-evolution-population-betax}
        \end{align}
\end{subequations}
Likewise, the population-level state evolution for both ${\alpha}_{\bm{h}_i^{t}}$ and ${\beta}_{\bm{h}_i^{t}}$:
\begin{subequations}\label{subeq:population-iterative_h}
        \begin{align}
        \alpha_{\bm{h}_i^{t+1}} & =(1-\eta) \alpha_{\bm{h}_i^{t}}+\eta\frac{{q_i\alpha}_{\bm{x}_i^{t}}}{{\alpha}_{\bm{x}_i^{t}}^2+{\beta}_{\bm{x}_i^{t}}^2},\label{eq:state-evolution-population-alphah}\\
        \beta_{\bm{h}_i^{t+1}} & =(1-\eta )\beta_{\bm{h}_i^{t}}.\label{eq:state-evolution-population-betah}
        \end{align}
\end{subequations}
In finite-sample case, the dynamics of the randomly initialized Wirtinger
flow iterates can be represented as
\begin{align}\label{eq:update_finite}
\bm{z}_i^{t+1 }=&\left[\begin{array}{c}
\bm{h}_i^{t+1} \\
\bm{x}_i^{t+1}
\end{array}\right]=\left[\begin{array}{c}
        \bm{h}_i^{t} -{\eta}/{\|\bm{x}_i^t\|_2^2}\cdot\nabla_{\bm{h}_i}F\left(\bm{z}\right) \\
        \bm{x}_i^{t} - {\eta}/{\|\bm{x}_i^t\|_2^2}\cdot\nabla_{\bm{x}_i}F\left(\bm{z}\right) 
        \end{array}\right]-\notag\\
        &\quad-\left[\begin{array}{c}
        {\eta}/{\|\bm{x}_i^t\|_2^2}
        \cdot( \nabla_{\bm{h}_i}f\left(\bm{z}\right) - \nabla_{\bm{h}_i}F\left(\bm{z}\right)) \\
        {\eta}/{\|\bm{h}_i^t\|_2^2}
        \cdot (\nabla_{\bm{x}_i}f\left(\bm{z}\right) - \nabla_{\bm{x}_i}F\left(\bm{z}\right) )
        \end{array}\right].
\end{align}
Under the assumption that the last term in (\ref{eq:update_finite}) is well-controlled, which will be justified in Appendix \ref{sec:proof_state_evo}, we arrive at the  approximate state evolution:
\begin{subequations}\label{eq:app_state}
\begin{align}
        \alpha_{\bm{h}_i^{t+1}} &= (1-\eta+\frac{\eta q_i \psi_{\bm{h}_i^{t}}}{{\alpha}_{\bm{x}_i^{t}}^2+{\beta}_{\bm{x}_i^{t}}^2}) \alpha_{\bm{h}_i^{t}}+\eta(1-\rho_{\bm{h}_i^{t}})\frac{{q_i\alpha}_{\bm{x}_i^{t}}}{{\alpha}_{\bm{x}_i^{t}}^2+{\beta}_{\bm{x}_i^{t}}^2},\label{eq:evo_alpha_H}\\
        \beta_{\bm{h}_i^{t+1}} & =(1-\eta+\frac{\eta q_i\varphi_{{\bm{h}_i^{t}}}}{{\alpha}_{\bm{x}_i^{t}}^2+{\beta}_{\bm{x}_i^{t}}^2} )\beta_{\bm{h}_i^{t}},\\
\alpha_{\bm{x}_i^{t+1}} & =(1-\eta+\frac{\eta q_i \psi_{\bm{x}_i^{t}}}{{\alpha}_{\bm{h}_i^{t}}^2+{\beta}_{\bm{h}_i^{t}}^2}) \alpha_{\bm{x}_i^{t}}+\eta(1-\rho_{\bm{x}_i^{t}})\frac{{q_i\alpha}_{\bm{h}_i^{t}}}{{\alpha}_{\bm{h}_i^{t}}^2+{\beta}_{\bm{h}_i^{t}}^2},\label{eq:evo_alpha_X}\\
\beta_{\bm{x}_i^{t+1}} & =(1-\eta+\frac{\eta q_i\varphi_{{\bm{x}_i^{t}}}}{{\alpha}_{\bm{h}_i^{t}}^2+{\beta}_{\bm{h}_i^{t}}^2} )\beta_{\bm{x}_i^{t}},\label{eq:evo_beta_X}
\end{align}
\end{subequations}
where $\{\psi_{\bm{h}_i^{t}}\},\{\psi_{\bm{x}_i^{t}}\},\{\varphi_{\bm{h}_i^{t}}\} ,\{\varphi_{\bm{x}_i^{t}}\},\{\rho_{\bm{h}_i^{t}}\}$ and $\{\rho_{\bm{x}_i^{t}}\}$ represent the perturbation terms.
\subsection{Dynamics of Approximate State Evolution}
To begin with, we define the discrepancy between the estimate $\bm{z}$ and the ground truth $\bar{\bm{z}}$ as the distance function, given as
\begin{align}\label{dist}
\mbox{dist}(\bm{z},\bar{\bm{z}})=\left(\sum_{i=1}^s\mbox{dist}^2(\bm{z}_i,\bar{\bm{z}}_i)\right)^{1/2},
\end{align}
where $\mbox{dist}^2(\bm{z}_i,\bar{\bm{z}}_i) = \min\limits_{\alpha_i\in\mathbb{ C}}({{\|\frac{1}{{\alpha_i}^*}\bm{h}_i-\bar{\bm{h}}_i \|_2^2+\|\alpha_i \bm{x}_i - \bar{\bm{x}}_i\|_2^2 }})/{d_i}$ for $i = 1,\cdots,s$. Here, $d_i = \|\bar{\bm{h}}_i\|_2^2+\|\bar{\bm{x}}_i\|_2^2$ and each $\alpha_i$ is the alignment parameter. It is easily seen that if $\alpha_{\bm{h}_i^t}$ (\ref{eq:alpha_h}), $\beta_{\bm{h}_i^t}$ (\ref{eq:beta_h}), $\alpha_{\bm{x}_i^t}$ (\ref{eq:alpha_x}) and $\beta_{\bm{x}_i^t}$ (\ref{eq:beta_x}) obey 
\begin{align}\label{eq:condition_alpha_beta}
&|\alpha_{\bm{h}_i^t}-q_i|\leq \frac{\gamma}{2\kappa\sqrt{s}}\quad\text{and}\quad\beta_{\bm{h}_i^t}\leq\frac{\gamma}{2\kappa\sqrt{s}}\quad\text{and}\notag\\&|\alpha_{\bm{x}_i^t}-q_i|\leq \frac{\gamma}{2\kappa\sqrt{s}}\quad\text{and}\quad\beta_{\bm{x}_i^t}\leq\frac{\gamma}{2\kappa\sqrt{s}},
\end{align} for $i = 1,\cdots, s$, then 
$
\mbox{dist}(\bm{z},\bar{\bm{z}})
\leq\gamma$.
Moreover, based triangle inequality, there is 
$
\mathrm{error}(\bm{\theta},\bar{\bm{\theta}})\leq\mbox{dist}(\bm{z},
\bar{\bm{z}})\leq \gamma
$.

In this subsection, we shall show that as long as the approximate state evolution (\ref{eq:app_state}) holds, there exists some constant $T_{\gamma} = \co(s\log \max{\{K,N\}})$ satisfying condition (\ref{eq:condition_alpha_beta}). This is demonstrated in the following Lemma.
Prior to that, we first list several conditions and definitions that contribute to the lemma.
\begin{itemize}
        \item The initial points obey 
        \begin{subequations}\label{eq:init_condition}
                \begin{align}
                        &\alpha_{\bm{h}_i^0}\geq\frac{q_i}{K\log K}\quad\text{and}
                        \quad\alpha_{\bm{x}_i^0}\geq\frac{q_i}{N\log N},\\
                        &\sqrt{\alpha_{\bm{h}_i^0}^{2}+\beta_{\bm{h}_i^0}^{2}}\in\bigg[1-\frac{1}{\log K},1+\frac{1}{\log K}\bigg]q_i,\\
                        &
                        \sqrt{\alpha_{\bm{x}_i^0}^{2}+\beta_{\bm{x}_i^0}^{2}}\in\bigg[1-\frac{1}{\log N},1+\frac{1}{\log N}\bigg]q_i,
                \end{align}
                for $i = 1,\cdots,s$.
        \end{subequations}
        \item Define
        \begin{align}\label{eq:def-T-gamma}
                T_{\gamma} & :=\min\big\{ t: \text{satifes (\ref{eq:condition_alpha_beta})}\big\},
        \end{align} where $\gamma>0$ is some sufficiently small constant.
        \item Define 
        \begin{align}
                T_{1} & :=\min\bigg\{ t: \min_i\frac{\alpha_{\bm{h}_i^t}}{q_i}\geq \frac{c_{7}}{\log^{5}m},\notag\\&\qquad\qquad\qquad\qquad\qquad\min_i\frac{\alpha_{\bm{x}_i^t}}{q_i}\geq \frac{c_{7}^\prime}{\log^{5}m}\bigg\} ,\label{eq:defn-T1}\\
                T_{2} & :=\min\left\{ t:\min_i \frac{\alpha_{\bm{h}_i^t}}{q_i}> c_{8},~ \min_i\frac{\alpha_{\bm{x}_i^t}}{q_i}> c_{8}^\prime\right\} ,\label{eq:defn-T2}
        \end{align}
        for some small absolute positive constants $c_{7},c_7^\prime,c_8,c_8^\prime>0$. 
        \item For $0\leq t\leq T_{\gamma}$, it has 
        \begin{align}
                &\frac{1}{2\sqrt{K\log K}}\leq\frac{\alpha_{\bm{h}_i^t}}{q_i}\leq2,~ c_{5}\leq\frac{\beta_{\bm{h}_i^t}}{q_i}\leq1.5\quad\text{and}\notag\\
                &\frac{\alpha_{\bm{h}_i^{t+1}}/\alpha_{\bm{h}_i^t}}{\beta_{\bm{h}_i^{t+1}}/\beta_{\bm{h}_i^t}}\geq1+c_{5}\eta,~ i = 1,\cdots,s,\label{eq:alpha-beta-range-SNRh}\\
                &\frac{1}{2\sqrt{N\log N}}\leq\frac{\alpha_{\bm{x}_i^t}}{q_i}\leq2,~ c_{6}\leq\frac{\beta_{\bm{x}_i^t}}{q_i}\leq1.5\quad\text{and}\notag\\
                &\frac{\alpha_{\bm{x}_i^{t+1}}/\alpha_{\bm{x}_i^t}}{\beta_{\bm{x}_i^{t+1}}/\beta_{\bm{x}_i^t}}\geq1+c_{6}\eta,~ i = 1,\cdots,s,\label{eq:alpha-beta-range-SNRx}
        \end{align}for some constants $c_{5},c_{6}>0$.
\end{itemize}
\begin{lemma}\label{lemma:iterative} Assume that
        the initial points obey condition (\ref{eq:init_condition}) and the perturbation terms in the approximate state evolution (\ref{eq:app_state}) obey
        $
        \max\left\{ |\psi_{\bm{h}_i^{t}}|,|\psi_{\bm{x}_i^{t}}|,| \varphi_{\bm{h}_i^{t}}|, |\varphi_{\bm{x}_i^{t}}| , \absn{\rho_{\bm{x}_i^{t}}}\right\} \leq\frac{c}{\log m},
        $
        for $i = 1,\cdots,s $, $t=0,1,\cdots$ and some sufficiently small constant $c>0$. 
        \begin{enumerate}
                
                \item Then for any sufficiently large $K,N$ and the stepsize $\eta>0$ that obeys $\eta\asymp s^{-1}$, it follows
               $
                        T_{\gamma}\lesssim s\log (\max{\{K,N\}})\label{eq:Tgamma-UB}
$ and
                (\ref{eq:alpha-beta-range-SNRh}), (\ref{eq:alpha-beta-range-SNRx}).
                \item Then with the stepsize $\eta>0$ following $\eta\asymp s^{-1}$, one has that
                $T_{1}\leq T_{2}\leq T_{\gamma}\lesssim s\log \max\{K,N\}$, $T_{2}-T_{1}\lesssim s\log\log m$,
                $T_{\gamma}-T_{2}\lesssim s$.

        \end{enumerate}
\end{lemma}
\begin{proof}
The proof of Lemma \ref{lemma:iterative} is inspired by the proof of Lemma 1 in \cite{chen2018}.
\end{proof}
The random initialization (\ref{eq:initial}) satisfies the condition (\ref{eq:init_condition}) with probability at least $ 1-\co(1/\sqrt{\log \min\{K,N\}})$ \cite{chen2018}.   According to this fact, Lemma \ref{lemma:iterative} ensures that under both random initialization (\ref{eq:initial}) and approximate state evolution (\ref{eq:app_state}) with the stepsize $\eta\asymp s^{-1}$, Stage I only lasts a few iterations, i.e., $T_\gamma = \co(s\log \max\{K,N\})$.  In addition, Lemma \ref{lemma:iterative} demonstrates the exponential growth of the ratios, i.e., ${\alpha_{\bm{h}_i^{t+1}}/\alpha_{\bm{h}_i^t}},{\beta_{\bm{h}_i^{t+1}}/\beta_{\bm{h}_i^t}}$, which contributes to the short duration of Stage I.

Moreover,
Lemma \ref{lemma:iterative} defines the midpoints $T_1$ when the sizes of the signal component, i.e., $\alpha_{\bm{h}_i^t}$ and $\alpha_{\bm{x}_i^t}$, $i = 1,\cdots, s$, become sufficiently large, which is crucial to the following analysis. In particular, when establishing the approximate state evolution (\ref{eq:app_state}) in Stage I, we analyze two subphases of Stage I individually:
\begin{itemize}
\item Phase 1:  consider the iterations in $0\leq t\leq T_1$, \item Phase 2: consider the iterations in $T_1<t\leq T_\gamma$,
\end{itemize} 
where $T_1$ is defined in (\ref{eq:defn-T1}).
\subsection{Leave-one-out Approach}\label{key}
According to Section \ref{section:dym} and Lemma \ref{lemma:iterative}, the unique challenge in  establishing the approximate state evolution (\ref{eq:app_state}) is to bound the perturbation terms to certain order, i.e., $ |\psi_{\bm{h}_i^{t}}|,|\psi_{\bm{x}_i^{t}}|,| \varphi_{\bm{h}_i^{t}}|, |\varphi_{\bm{x}_i^{t}}| ,\absn{\rho_{\bm{h}_i^{t}}} , \absn{\rho_{\bm{x}_i^{t}}} \ll 1/{\log m}$ for $i = 1,\cdots,s$. To achieve this goal, we exploit some variants of leave-one-out sequences \cite{chen2018,dong2018} to establish the ``near-independence" between $\{\bm{z}_i^t\}$ and $\{\bm{a}_i\}$. Hence, some terms can be approximated by a sum of independent variables with well-controlled weight, thereby be controlled via central limit theorem.

In the following, we define three sets of auxiliary sequences  $\{\bm{z}^{t,(l)}\}$, $\{\bm{z}^{t,\sgn}\}$ and $\{\bm{z}^{t,\sgn,(l)}\}$, respectively.
\begin{itemize}
        \item \emph{Leave-one-out sequences} $\{\bm{z}^{t,(l)}\}_{t\geq0}$. For
        each $1\leq l\leq m$, the  auxiliary sequence $\{\bm{z}^{t,(l)}\}$ is established by dropping the $l$-th sample and runs randomly initialized Wirtinger
flow with objective function
        \begin{equation}
        f^{(l)}\left(\bm{z}\right)=\sum_{j:j\neq l}\Big|\sum_{i=1}^{ s}\bm{b}_j^{\mathsf{H}}\bm{h}_i\bm{x}_i^{\mathsf{H}}\bm{a}_{ij}-{y}_j\Big|^2.\label{eq:auxiliary-loss-LOO}
        \end{equation}
        Thus, the sequences $\{\bm{z}_i^{t,(l)}\}$ (recall the definition of $\bm{z}_i$ (\ref{eq:z})) are statistically independent of $\{\bm{a}_{il}\}$. 
        \item \emph{Random-sign sequences} $\left\{ \bm{z}^{t,\mathrm{sgn}}\right\} _{t\geq0}$.
Define the auxiliary design vectors $\left\{ \bm{a}_{ij}^{\mathrm{sgn}}\right\} $ as
        \begin{equation}
        \bm{a}_{ij}^{\mathrm{\mathrm{sgn}}}:=\left[\begin{array}{c}
        \xi_{ij}a_{ij,1}\\
        \bm{a}_{ij,\perp}
        \end{array}\right],\label{eq:auxiliary-random-sign-vector}
        \end{equation}
        where $\left\{ \xi_{ij}\right\}$ is
        a set of standard complex uniform random variables independent of $\left\{ \bm{a}_{ij}\right\} $, i.e.,
\begin{align}\label{eq:xi}
        \xi_{ij}\overset{\text{i.i.d.}}{=}{u}/{|u|},  
\end{align}
where $u\sim \mathcal{N}({0},\frac{1}{2})+i\mathcal{N}({0},\frac{1}{2})$.
Moreover, with the corresponding $\xi_{ij}$, the auxiliary design vector $\{\bm{b}_j^{\sgn}\}$ is defined as $\bm{b}_j^{\sgn} = \xi_{ij}\bm{b}_j$. With these auxiliary design vectors, the sequences $\{\bm{z}^{t,\mathrm{sgn}}\}$ are generated by
running randomly initialized Wirtinger
flow with respect to the loss function
        \begin{equation}
        f^{\mathrm{sgn}}\big(\bm{z}\big)=\sum_{j=1}^m\Big|\sum_{i=1}^{ s}\bm{b}_j^{\sgn\, H}\bm{h}_i\bm{x}_i^{\mathsf{H}}\bm{a}_{ij}^{\sgn}-\bm{b}_j^{\sgn\,H}\bar{\bm{h}}_i\bar{\bm{x}}_i^{\mathsf{H}}\bm{a}_{ij}^{\sgn}\Big|^2.\label{eq:f-sgn-LOO}
        \end{equation}
Note that these auxiliary design
        vectors, i.e., $\{\bm{a}_{ij}^{\mathrm{\mathrm{sgn}}}\},\{\bm{b}_j^{\sgn}\}$ produce the same measurements as $\left\{ \bm{a}_{ij}\right\}, \{\bm{b}_j\}$:
        \begin{equation}
        \bm{b}_j^{\sgn \mathsf{H}}\bar{\bm{h}}_i\bar{\bm{x}}_i^{\mathsf{H}}\bm{a}_{ij}^{\sgn}=\bm{b}_j^{\mathsf{H}}\bar{\bm{h}}_i\bar{\bm{x}}_i^{\mathsf{H}}\bm{a}_{ij}=q_ia_{ij,1}\bm{b}_j^{\mathsf{H}}\bar{\bm{h}}_i ,\label{eq:by-construction}
        \end{equation}for $1\leq i\leq s,1\leq j\leq m$.
\end{itemize}
Note that all the auxiliary sequences are assumed to have the same initial point, namely, for $
1\leq l\leq m$,
\begin{align}\label{eq:seq_initial}
\{\bm{z}^0\} = \{\bm{z}^{0,(l)}\}  = \{\bm{z}^{0,\sgn}\} = \{\bm{z}^{0,\sgn,(l)}\}.
\end{align}

In view of the ambiguities, i.e., $\bar{\bm{h}}_i\bar{\bm{x}}_i = \frac{1}{{\omega}^*}\bar{\bm{h}}_i(\omega\bar{\bm{x}}_i)^{\mathsf{H}},$  several alignment parameters are further defined for the sequel analysis.  Specifically, the alignment parameter between ${\bm{z}}_i^{t,(l)} = [{\bm{h}}_i^{t,(l)\top}~~{\bm{x}}_i^{t,(l)\top}]^\top$ and $\widetilde{\bm{z}}_i^t = [\widetilde{\bm{h}}_i^{t\,\top}~~\widetilde{\bm{x}}_i^{t\,\top}]^\top$, where $\widetilde{\bm{h}}^t_i = \frac{1}{{\omega_i^t}^*}\bm{h}_i^t$ and $\widetilde{\bm{x}} ^t_i= \omega_i^t\bm{x}_i^t$, are represented as
\begin{align}
\omega_{i,\text{mutual}}^{t,(l)}: = \argmin_{\omega\in\mathbb{C}}\left\| \frac{1}{{\omega}^*}\bm{h}_i^{t,(l)}\!-\!\frac{1}{{\omega_i^t}^*}\bm{h}_i^{t}\right\|_2^2\!\!+\! \left\|\omega\bm{x}_i^{t,(l)}\!\!-\!\omega_i^t\bm{x}_i^{t}\right\|_2^2,
\end{align}
 for $i = 1,\cdots, s$. In addition, we denote $\widehat{\bm{z}}_i^{t,(l)} = [\widehat{\bm{h}}_i^{t,(l)\, \top}~\widehat{\bm{x}}_i^{t,(l)\, \top}]^\top$ where
\begin{align}\label{eq:mutual_seq}
\widehat{\bm{h}}_i^{t,(l)}:=\frac{1}{({\omega_{i,\text{mutual}}^{t,(l)}})^*}\bm{h}_i^{t,(l)}\quad\text{and}\quad\widehat{\bm{x}}_i^{t,(l)}:={{\omega_{i,\text{mutual}}^{t,(l)}}}\bm{x}_i^{t,(l)}.
\end{align}
Define the alignment parameter between ${\bm{z}}_i^{t,\mathrm{sgn}} = [{\bm{h}}_i^{t,\sgn\,\top}~~{\bm{x}}_i^{t,\sgn\,\top}]^\top$ and ${\bm{z}}_i^t = [{\bm{h}}_i^{t\,\top}~~{\bm{x}}_i^{t\,\top}]^\top$ as
\begin{align}
\omega_{i,\sgn}^{t}: =  \argmin_{\omega\in\mathbb{C}}\left\| \frac{1}{{\omega}^*}\bm{h}_i^{t,\sgn}\!-\!\frac{1}{{\omega_i^t}^*}\bm{h}_i^{t}\right\|_2^2\!+\! \left\|\omega\bm{x}_i^{t,\sgn}\!-\!\omega_i^t\bm{x}_i^{t}\right\|_2^2,
\end{align}
for $i = 1,\cdots, s$. In addition, we denote $\check{\bm{z}}_i^{t,\mathrm{sgn}} = [\check{\bm{h}}_i^{t,\sgn\,\top}~\check{\bm{x}}_i^{t,\sgn\,\top}]^\top$ where
\begin{align}
\check{\bm{h}}_i^{t,\mathrm{sgn}}:=\frac{1}{({\omega_{i,\sgn}^{t}})^*}
\bm{h}_i^{t,\mathrm{sgn}}\quad\text{and}\quad\check{\bm{x}}_i^{t,\mathrm{sgn}}:={{\omega_{i,\sgn}^{t}}}\bm{x}_i^{t,\mathrm{sgn}}.
\end{align}
\subsection{Establishing Approximate State Evolution for Phase 1 of Stage I}
In this subsection, we will justify that the approximate state evolution (\ref{eq:app_state}) for both the size of the signal component and the size of the perpendicular component is satisfied during Phase I.  In particular, we establish a collection of induction hypotheses which are crucial to the justification of approximate state evolution (\ref{eq:app_state}), and then identify these hypotheses via inductive argument.

To begin with, we list all the induction hypotheses: for $1\leq i \leq s$,
\begin{subequations}\label{subeq:induction}
        \begin{align}
                &\max_{ 1\leq l\leq m}\dist\paren{{\bm{z}}_i^{t,(l)},\widetilde{\bm{z}}_i^{t}}\notag\\
                \leq&(\beta_{{\bm{h}}_i^{t}}+\beta_{{\bm{x}}_i^{t}})\left(1+\frac{1}{s\log m}\right)^{t}C_{1}\frac{ s\mu^2\kappa\sqrt{\max\{K,N\}\log^{8}m}}{m}\label{eq:h-induction-leave}\\
&\max_{ 1\leq l\leq m}\dist\paren{      \bar{\bm{h}}_i^{\mathsf{H}}\bm{h}_i^{t,(l)},\bar{\bm{h}}_i^{\mathsf{H}}\wt{\bm{h}}_i^{t}}\cdot\normn{\bar{\bm{h}}_i}{2}^{-1}\notag\\
\leq&\alpha_{{\bm{h}}_i^{t}} \left(1+\frac{1}{s\log m}\right)^{t}C_{2}\frac{s\mu^2\kappa\sqrt{K\log^{13}m}}{m}\label{eq:l_h}\\
        &\max_{ 1\leq l\leq m}\dist\paren{      {{x}}_{i1}^{t,(l)},\widetilde{{x}}_{i1}^{t}}\notag\\
        \leq&\alpha_{{\bm{x}}_i^{t}} \left(1+\frac{1}{s\log m}\right)^{t}C_{2}\frac{s\mu^2\kappa\sqrt{N\log^{13}m}}{m}\label{eq:l_x}\\
        &\max_{1\leq i\leq s}\dist{\paren{\bm{h}_i^{t,\mathrm{sgn}},\widetilde{\bm{h}}_i^{t}}}\notag\\
        \leq&\alpha_{{\bm{h}}_i^t}\left(1+\frac{1}{s\log m}\right)^{t}C_{3}\sqrt{\frac{s\mu^2\kappa^2K\log^{8}m}{m}}\label{eq:sgn_h}\\
        &\max_{1\leq i\leq s}\dist{\paren{\bm{x}_i^{t,\mathrm{sgn}},\wt{\bm{x}}_i^{t}}}\notag\\
        \leq&\alpha_{{\bm{x}}_i^t}\left(1+\frac{1}{s\log m}\right)^{t}C_{3}\sqrt{\frac{s\mu^2\kappa^2N\log^{8}m}{m}}\label{eq:sgn_x}\\
        &\max_{ 1\leq l\leq m}\left\Vert \widetilde{\bm{h}}_{i}^{t}-\widehat{\bm{h}}_{i}^{t,\left(l\right)}-\widetilde{\bm{h}}_{i}^{t,\sgn}+\widehat{\bm{h}}_{i}^{t,\sgn,\left(l\right)}\right\Vert _{2}\notag\\
 \leq   &\alpha_{{\bm{h}}_i^t}\left(1+\frac{1}{s\log m}\right)^{t}C_{4}\frac{s\mu^2\sqrt{K\log^{16}m}}{m},\label{eq:h-induction-double}\\
        &\max_{ 1\leq l\leq m}\left\Vert \widetilde{\bm{x}}_{i}^{t}-\widehat{\bm{x}}_{i}^{t,\left(l\right)}-\widetilde{\bm{x}}_{i}^{t,\sgn}+\widehat{\bm{x}}_{i}^{t,\sgn,\left(l\right)}\right\Vert _{2}\notag\\
        \leq    &\alpha_{{\bm{x}}_i^t}\left(1+\frac{1}{s\log m}\right)^{t}C_{4}\frac{s\mu^2\sqrt{N\log^{16}m}}{m},\label{eq:x-induction-double}\\
&       c_{5}\leq  \left\Vert \bm{h}_i^{t}\right\Vert _{2},\left\Vert \bm{x}_i^{t}\right\Vert _{2}\leq C_{5},\label{eq:induction-norm-size}\\
&       \left\Vert \bm{h}_i^{t}\right\Vert _{2}  \leq5\alpha_{\bm{h}_i^{t}}\sqrt{\log^5 m},\label{eq:h-induction-norm-relative}\\
&       \left\Vert \bm{x}_i^{t}\right\Vert _{2}  \leq5\alpha_{\bm{x}_i^{t}}\sqrt{\log^5 m},\label{eq:x-induction-norm-relative}
        \end{align}\end{subequations}
        where $C_{1},\cdots,C_{5}$ and $c_{5}$ are some absolute positive
        constants and $\widehat{\bm{x}}_{i},\wt{\bm{x}}_{i}, \widehat{\bm{h}}_{i},\wt{\bm{h}}_{i}$ are defined in Section \ref{key}.

 Specifically, (\ref{eq:h-induction-leave}), (\ref{eq:l_x}), (\ref{eq:sgn_h}) and (\ref{eq:sgn_x}) identify that the auxiliary sequences  $\{\bm{z}^{t,(l)}\}$ and $\{\bm{z}^{t,\sgn}\}$ are extremely close to the original sequences $\{\bm{z}^{t}\}$. In addition, as claimed in (\ref{eq:h-induction-double}) and (\ref{eq:x-induction-double}), $\widetilde{\bm{h}}_{i}^{t}-\widetilde{\bm{h}}_{i}^{t,\sgn}$ (resp. $\widetilde{\bm{x}}_{i}^{t}-\widetilde{\bm{x}}_{i}^{t,\sgn}$) and $\widehat{\bm{h}}_{i}^{t,\left(l\right)}-\widehat{\bm{h}}_{i}^{t,\sgn,\left(l\right)}$ (resp. $\widehat{\bm{x}}_{i}^{t,\left(l\right)}-\widehat{\bm{x}}_{i}^{t,\sgn,\left(l\right)}$) are also exceedingly close to each other. The hypotheses (\ref{eq:induction-norm-size}) illustrates that the norm of the iterates $\{\bm{h}_i^t\}$ (resp. $\{\bm{ x}_i^t\}$) is well-controlled in Phase 1. 
 Moreover,  (\ref{eq:h-induction-norm-relative})  (resp. (\ref{eq:x-induction-norm-relative}))
 indicates that $\alpha_{\bm{h}_i^{t}}$ (resp. $\alpha_{\bm{x}_i^{t}}$) is comparable to $\norm{\bm{h}_i^t}{2}$ (resp. $\norm{\bm{x}_i^t}{2}$).

We are moving to prove that if the induction hypotheses (\ref{subeq:induction}) hold for the $t$-th iteration, then $\alpha_{\bm{h}_i}$ (\ref{eq:state-evolution-population-alphah}), $\beta_{\bm{h}_i}$ (\ref{eq:state-evolution-population-betah}), $\alpha_{\bm{x}_i}$ (\ref{eq:state-evolution-population-alphax}) and $\beta_{\bm{x}_i}$ (\ref{eq:state-evolution-population-betax}) obey the approximate state evolution (\ref{eq:app_state}). This is demonstrated in Lemma \ref{lemma:state_evo}.

 \begin{lemma}\label{lemma:state_evo}Suppose $m\geq Cs^2\mu^2\max\{K,N\}\log^{10}m$
        for some sufficiently large constant $C>0$. For any $0\leq t\leq T_{1}$  (\ref{eq:defn-T1}), if the $t$-th iterate satisfies
        the induction hypotheses (\ref{subeq:induction}) , then for $i = 1,\cdots,s$, with probability
        at least $1-c_1m^{-\nu}-c_1me^{-c_2N}$ for some constants $\nu, c_1, c_2>0$,
         the approximate evolution state (\ref{eq:app_state}) holds for  some $|\psi_{\bm{h}_i^{t}}|,|\psi_{\bm{x}_i^{t}}|,| \varphi_{\bm{h}_i^{t}}|, |\varphi_{\bm{x}_i^{t}}| ,\absn{\rho_{\bm{h}_i^{t}}}, \absn{\rho_{\bm{x}_i^{t}}}\ll1/\log m$, $i = 1,\cdots,s$. \end{lemma}
 \begin{proof}
        Please refer to Appendix \ref{sec:proof_state_evo} for details.
 \end{proof}    
In the sequel, we will prove the hypotheses (\ref{subeq:induction}) hold for Phase 1 of Stage I via inductive arguments. Before moving forward, we first investigate the incoherence between $\{\bm{x}_i^t\}$, $\{\bm{x}_i^{t,\sgn}\}$ (resp. $\{\bm{h}_i^t\}$, $\{\bm{h}_i^{t,\sgn}\}$) and $\{\bm{a}_{ij}\}$, $\{\bm{a}_{ij}^{\sgn}\}$ (resp. $\{\bm{b}_{j}\}$, $\{\bm{b}_{j}^{\sgn}\}$).

\begin{lemma}\label{lemma:consequence-main-text}
        Suppose that $m\geq Cs^2\mu^2\max\{K,N\}\log^8 m$
        for some sufficiently large constant $C>0$ and the $t$-th iterate
satisfies the induction hypotheses (\ref{subeq:induction}) for $t\leq T_{0}$ (\ref{eq:defn-T1}),
        then with probability at least $1-c_1m^{-\nu}-c_1me^{-c_2N}$ for some constants $\nu, c_1, c_2>0$,\begin{subequations}\label{eq:claim_1}
                \begin{align}
                \max_{1\leq i \leq s, 1\leq l\leq m}\left|\bm{a}_{il}^{\mathsf{H}}\wt{\bm{x}}_i^t \right|\cdot\|\wt{\bm{x}}_i^t\|_2^{-1}&\lesssim\sqrt{\log m},\\
                \max_{1\leq i \leq s, 1\leq l\leq m}\left|\bm{a}_{il,\perp}^{\mathsf{H}}\wt{\bm{x}}_{i\perp}^{t}\right|\cdot\|\wt{\bm{x}}_{i\perp}^{t}\|_2^{-1}&\lesssim\sqrt{\log m},\label{eq:x_perp}\\
                \max_{1\leq i \leq s, 1\leq l\leq m}\left|\bm{a}_{il}^{\mathsf{H}}\check{\bm{x}}_{i}^{t,\mathrm{sgn}}\right|\cdot\|\check{\bm{x}}_{i}^{t,\mathrm{sgn}}\|_2^{-1}&\lesssim\sqrt{\log m},\\
                \max_{1\leq i \leq s, 1\leq l\leq m}\left|\bm{a}_{il,\perp}^{\mathsf{H}}\check{\bm{x}}_{i\perp}^{t,\mathrm{sgn}}\right|\cdot\|\check{\bm{x}}_{i\perp}^{t,\mathrm{sgn}}\|_2^{-1}&\lesssim\sqrt{\log m},\label{eq:x_sgn_perp}\\
                \max_{1\leq i \leq s, 1\leq l\leq m}\left|\bm{a}_{il}^{\mathrm{sgn}\,\mathsf{H}}\check{\bm{x}}_i^{t,\mathrm{sgn}}\right|\cdot\|\check{\bm{x}}_i^{t,\mathrm{sgn}}\|_2^{-1}&\lesssim\sqrt{\log m},
                \end{align}
        \end{subequations}
        \begin{subequations}\label{eq:claim_2}
                \begin{align}
        &       \max_{1\leq i \leq s, 1\leq l\leq m}\left|\bm{b}_{l}^{\mathsf{H}}\wt{\bm{h}}_i^t \right|\cdot\|\wt{\bm{h}}_i^t\|_2^{-1}\lesssim\frac{\mu}{\sqrt{m}}\log^2m,\label{eq:wth}\\
        &       \max_{1\leq i \leq s, 1\leq l\leq m}\left|\bm{b}_{l}^{\mathsf{H}}\check{\bm{h}}_i^{t,\mathrm{sgn}}\right|\cdot\|\check{\bm{h}}_i^{t,\mathrm{sgn}}\|_2^{-1}\lesssim\frac{\mu}{\sqrt{m}}\log^2m,\\
        &       \max_{1\leq i \leq s, 1\leq l\leq m}\left|\bm{b}_{l}^{\mathrm{sgn}\,\mathsf{H}}\check{\bm{h}}_i^{t,\mathrm{sgn}}\right|\cdot\|\check{\bm{h}}_i^{t,\mathrm{sgn}}\|_2^{-1}\lesssim\frac{\mu}{\sqrt{m}}\log^2m.
                \end{align}
        \end{subequations}
        \end{lemma}
\begin{proof}
Based on the induction hypotheses (\ref{subeq:induction}),  we can prove the claim (\ref{eq:claim_1}) in Lemma \ref{lemma:consequence-main-text} by invoking the triangle inequality, Cauchy-Schwarz inequality and standard Gaussian concentration. Furthermore, based on the induction hypotheses (\ref{subeq:induction}), the claim (\ref{eq:claim_2}) can be identified according to the definition of the incoherence parameter in Definition \ref{inco} and the fact $\norm{\bm{b}_j}{2} =  \sqrt{K/M}$.
\end{proof}
Now we are ready to specify that the hypotheses (\ref{subeq:induction}) hold for $0\leq t\leq T_1$ (\ref{eq:defn-T1}). We aim to demonstrate that if the hypotheses (\ref{subeq:induction}) hold up to the $t$-th iteration for some $0\leq t\leq T_1$, then they hold for the $(t+1)$-th iteration. Since the case for $t  =0$ can be easily justified due to the equivalent initial points (\ref{eq:seq_initial}), we mainly focus the inductive step.

\begin{lemma}\label{lemma:xt-xt-l}Suppose the induction hypotheses
        (\ref{subeq:induction}) hold true up to the $t$-th iteration for
        some $t\leq T_{1}$ (\ref{eq:defn-T1}), then for $i = 1,\cdots,s$, with probability at least $1-c_1m^{-\nu}-c_1me^{-c_2N}$ for some constants $\nu, c_1, c_2>0$,
        \begin{align}
        &\max_{ 1\leq l\leq m}\dist\paren{{\bm{z}}_i^{t+1,(l)},\widetilde{\bm{z}}_i^{t+1}}\notag\\
        \leq&(\beta_{{\bm{h}}_i^{t+1}}+\beta_{{\bm{x}}_i^{t+1}})\left(1+\frac{1}{s\log m}\right)^{t+1}C_{1}\cdot\notag\\
        &\frac{ s\mu^2\kappa\sqrt{\max\{K,N\}\log^{8}m}}{m}
        \end{align}
        holds $m\geq Cs\mu^2\kappa\sqrt{\max\{K,N\}\log^{8}m}$ with
        some sufficiently large constant $C>0$ as long as the stepsize $\eta>0$ obeys $\eta\asymp s^{-1}$ and $C_{1}>0$
        is sufficiently large. \end{lemma}
In terms of the difference between $\bm{x}^{t}$ and $\bm{x}_i^{t,(l)}$ (resp. $\bm{h}_i^{t}$ and $\bm{h}_i^{t,(l)}$) along with the signal direction, i.e., (\ref{eq:l_h}) and (\ref{eq:l_x}), we reach the following lemma.

\begin{lemma}\label{lemma:xt-xt-l-signal}Suppose the induction hypotheses
        (\ref{subeq:induction}) hold true up to the $t$-th iteration for
        some $t\leq T_{1}$ (\ref{eq:defn-T1}), then with probability at least $1-c_1m^{-\nu}-c_1me^{-c_2N}$ for some constants $\nu, c_1, c_2>0$,
        \begin{align}
        &\max_{ 1\leq l\leq m}\dist\paren{      \bar{\bm{h}}_i^{\mathsf{H}}\bm{h}_i^{t+1,(l)},\bar{\bm{h}}_i^{\mathsf{H}}\wt{\bm{h}}_i^{t+1}}\cdot\normn{\bar{\bm{h}}_i}{2}^{-1}\notag\\
        \leq&\alpha_{{\bm{h}}_i^{t+1}} \left(1+\frac{1}{s\log m}\right)^{t+1}C_{2}\frac{s\mu^2\kappa\sqrt{K\log^{13}m}}{m}\\
        &\max_{ 1\leq l\leq m}\dist\paren{      {{x}}_{i1}^{t+1,(l)},\widetilde{{x}}_{i1}^{t+1}}\notag\\
        \leq&\alpha_{{\bm{x}}_i^{t+1}} \left(1+\frac{1}{s\log m}\right)^{t+1}C_{2}\frac{s\mu^2\kappa\sqrt{N\log^{13}m}}{m}\label{eq:dist_l_x}
        \end{align}
        holds for some sufficiently large $C_{2}>0$ with $C_2\gg C_4$, provided that $m\geq Cs\mu^2\kappa \max\{K,N\}\log^{12} m$ for some sufficiently large constant $C>0$ and the stepsize $\eta>0$ obeys $\eta\asymp s^{-1}$.
\end{lemma}
\begin{proof}
        Please refer to Appendix \ref{sec:proof_leave_one} for details.
\end{proof}
The next lemma concerns the relation between $\bm{h}_i^t$ and $\bm{h}_{i}^{t,\sgn}$, i.e., (\ref{eq:sgn_h}), and the relation between  $\bm{x}_i^t$ and $\bm{x}_i^{t,\sgn}$, i.e.,  (\ref{eq:sgn_x}).
\begin{lemma}\label{lemma:xt-xt-sgn}Suppose the induction hypotheses
        (\ref{subeq:induction}) hold true up to the $t$-th iteration for
        some $t\leq T_{1}$ (\ref{eq:defn-T1}), then with probability at least $1-c_1m^{-\nu}-c_1me^{-c_2N}$ for some constants $\nu, c_1, c_2>0$,
\begin{subequations}
        \begin{align}
        &\max_{1\leq i\leq s}\dist{\paren{\bm{h}_i^{t+1,\mathrm{sgn}},\widetilde{\bm{h}}_i^{t+1}}}\notag\\
        \leq&\alpha_{{\bm{h}}_i^{t+1}}\left(1+\frac{1}{s\log m}\right)^{t+1}C_{3}\sqrt{\frac{s\mu^2\kappa^2K\log^{8}m}{m}}\\
        &\max_{1\leq i\leq s}\dist{\paren{\bm{x}_i^{t+1,\mathrm{sgn}},\widetilde{\bm{x}}_i^{t+1}}}\notag\\
        \leq&\alpha_{{\bm{x}}_i^{t+1}}\left(1+\frac{1}{s\log m}\right)^{t+1}C_{3}\sqrt{\frac{s\mu^2\kappa^2N\log^{8}m}{m}}
        \end{align}
\end{subequations}
                holds for some sufficiently large $C_{3}>0$,
        provided that $m\geq Cs\mu^2\kappa^2\max\{K,N\}\log^{8} m$ for some sufficiently large constant $C>0$ and the stepsize $\eta>0$ obeys $\eta\asymp s^{-1}$. \end{lemma}
We still need to characterize the difference $\widetilde{\bm{h}}_{i}^{t}-\widehat{\bm{h}}_{i}^{t,\left(l\right)}-\widetilde{\bm{h}}_{i}^{t,\sgn}+\widehat{\bm{h}}_{i}^{t,,\sgn,\left(l\right)}$, i.e., (\ref{eq:h-induction-double}), and the difference $ \widetilde{\bm{x}}_{i}^{t}-\widehat{\bm{x}}_{i}^{t,\left(l\right)}-\widetilde{\bm{x}}_{i}^{t,\sgn}+\widehat{\bm{x}}_{i}^{t,\sgn,\left(l\right)}$, i.e.,  (\ref{eq:x-induction-double}), in the following lemma.
\begin{lemma}\label{lemma:double-diff}Suppose the induction hypotheses
        (\ref{subeq:induction}) hold true up to the $t$-th iteration for
        some $t\leq T_{1}$ (\ref{eq:defn-T1}), then with probability at least $1-c_1m^{-\nu}-c_1me^{-c_2N}$ for some constants $\nu, c_1, c_2>0$,
        \begin{subequations}
        \begin{align}
                &\max_{ 1\leq l\leq m}\left\Vert \widetilde{\bm{h}}_{i}^{t+1}-\widehat{\bm{h}}_{i}^{t+1,\left(l\right)}-\widetilde{\bm{h}}_{i}^{t+1,\sgn}+\widehat{\bm{h}}_{i}^{t+1,\sgn,\left(l\right)}\right\Vert _{2}\notag\\
                \leq&\alpha_{{\bm{h}}_i^{t+1}}\left(1+\frac{1}{s\log m}\right)^{t+1}C_{4}\frac{s\mu^2\sqrt{K\log^{16}m}}{m}\label{eq:h-l-sgn}
                        \\&\max_{ 1\leq l\leq m}\left\Vert \widetilde{\bm{x}}_{i}^{t+1,}-\widehat{\bm{x}}_{i}^{t+1,,\left(l\right)}-\widetilde{\bm{x}}_{i}^{t+1,\sgn}+\widehat{\bm{x}}_{i}^{t+1,\sgn,\left(l\right)}\right\Vert _{2}\notag\\
                \leq&\alpha_{{\bm{x}}_i^{t+1}}\left(1+\frac{1}{s\log m}\right)^{t+1}C_{4}\frac{s\mu^2\sqrt{N\log^{16}m}}{m}\label{eq:x-l-sgn}
                \end{align}
        \end{subequations}
        holds for some sufficiently large $C_{4}>0$, provided that $m\geq Cs\mu^2\max\{K,N\}\log^8 m$ for some sufficiently large constant $C>0$ and the stepsize $\eta>0$ obeys $\eta\asymp s^{-1}$. \end{lemma}
\begin{remark}
The arguments applied to prove Lemma \ref{lemma:xt-xt-l}-Lemma \ref{lemma:double-diff} are similar to each other. We thus mainly focus on the proof of (\ref{eq:dist_l_x}) in Lemma \ref{lemma:xt-xt-l-signal} in Appendix \ref{sec:proof_leave_one}.
\end{remark}
\subsection{Establishing Approximate State Evolution for Phase 2 of Stage I}
In this subsection, we move to prove that the approximate state evolution (\ref{eq:app_state}) holds for $T_1< t\leq T_\gamma$ ($T_\gamma$ and $T_1$ are defined in (\ref{eq:def-T-gamma}) and (\ref{eq:defn-T1}) respectively) via inductive argument. Different from the analysis in Phase 1, only $\{\bm{z}^{t,(l)}\}$ is sufficient to establish the ``near-independence" between iterates and design vectors when the sizes of the signal component follow $\alpha_{\bm{h}_i^t}$, $\alpha_{\bm{x}_i}\gtrsim{1}/{\log m}$ in Phase 2 (according to the definition of $T_1$).
As in Phase 1, we begin with specifying the induction hypotheses: for $1\leq i \leq s$,
\begin{subequations}\label{eq:hypo}
\begin{align}
&\max_{ 1\leq l\leq m}\dist\paren{{\bm{z}}_i^{t,(l)},\widetilde{\bm{z}}_i^{t}}\notag\\
\leq&(\beta_{{\bm{h}}_i^{t}}+\beta_{{\bm{x}}_i^{t}})\left(1+\frac{1}{s\log m}\right)^{t}C_{6}\frac{ s\mu^2\kappa\sqrt{\max\{K,N\}\log^{18}m}}{m}\label{eq:h-induction-leave_2}\\
&       c_{5}\leq  \left\Vert \bm{h}_i^{t}\right\Vert _{2},\left\Vert \bm{x}_i^{t}\right\Vert _{2}\leq C_{5},\label{eq:induction-norm-size_2}
\end{align}
\end{subequations}
From (\ref{eq:hypo}), we can conclude that one has
\begin{align}
        \max_{1\leq i \leq s, 1\leq l\leq m}\left|\bm{a}_{il}^{\mathsf{H}}\wt{\bm{x}}_i^t \right|\cdot\|\wt{\bm{x}}_i^t\|_2^{-1}&\lesssim\sqrt{\log m},\label{eq:incohx}\\
                \max_{1\leq i \leq s, 1\leq l\leq m}\left|\bm{b}_{l}^{\mathsf{H}}\wt{\bm{h}}_i^t \right|\cdot\|\wt{\bm{h}}_i^t\|_2^{-1}&\lesssim\frac{\mu}{\sqrt{m}}\log^2m,\label{incohh}
\end{align}
with probability at least $1-c_1m^{-\nu}-c_1me^{-c_2N}$ for some constants $\nu, c_1, c_2>0$ during $T_1<t\leq T_\gamma$ as long as $m\gg  Cs\mu^2\kappa K\log^{8}m$.

We then move to prove that if the induction hypotheses (\ref{subeq:induction}) hold for the $t$-th iteration, then $\alpha_{\bm{h}_i}$ (\ref{eq:state-evolution-population-alphah}), $\beta_{\bm{h}_i}$ (\ref{eq:state-evolution-population-betah}), $\alpha_{\bm{x}_i}$ (\ref{eq:state-evolution-population-alphax}) and $\beta_{\bm{x}_i}$ (\ref{eq:state-evolution-population-betax}) obey the approximate state evolution (\ref{subeq:induction}). This is demonstrated in Lemma \ref{lemma:state_evo_2}.

\begin{lemma}\label{lemma:state_evo_2}Suppose $m\geq Cs^2\mu^2\kappa^4\max\{K,N\}\log^{12}m$
for some sufficiently large constant $C>0$. For any $T_{1}\leq t\leq T_{\gamma}$  ($T_1$ and $T_\gamma$ are defined in (\ref{eq:def-T-gamma}) and (\ref{eq:defn-T1}) respectively), if the $t$-th iterate satisfies
the induction hypotheses (\ref{subeq:induction}) , then for $i = 1,\cdots,s$, with probability
at least $1-c_1m^{-\nu}-c_1me^{-c_2N}$ for some constants $\nu, c_1, c_2>0$,
the approximate evolution state (\ref{eq:app_state}) hold for  some $|\psi_{\bm{h}_i^{t}}|,|\psi_{\bm{x}_i^{t}}|,| \varphi_{\bm{h}_i^{t}}|, |\varphi_{\bm{x}_i^{t}}| ,\absn{\rho_{\bm{h}_i^{t}}}, \absn{\rho_{\bm{x}_i^{t}}}\ll1/\log m$, $i = 1,\cdots,s$. \end{lemma}
It remains to proof the induction step on the difference between leave-one-out sequences $\{\bm{z}^{t,(l)}\}$ and the original sequences $\{\bm{z}^t\}$, which is demonstrated in the following lemma.
\begin{lemma}\label{lemma:xt-xt-l_2}Suppose the induction hypotheses
        (\ref{subeq:induction}) are valid during Phase 1 and the induction hypotheses (\ref{eq:hypo}) hold true from $T_1$-th to the $t$-th for some $t\leq T_\gamma$ (\ref{eq:def-T-gamma}), then for $i = 1,\cdots,s$, with probability at least $1-c_1m^{-\nu}-c_1me^{-c_2N}$ for some constants $\nu, c_1, c_2>0$,
        \begin{align}
        &\max_{ 1\leq l\leq m}\dist\paren{{\bm{z}}_i^{t,(l)},\widetilde{\bm{z}}_i^{t}}\notag\\
        \leq&(\beta_{{\bm{h}}_i^{t+1}}+\beta_{{\bm{x}}_i^{t+1}})\left(1+\frac{1}{s\log m}\right)^{t+1}C_{6}\frac{ s\mu^2\kappa\sqrt{K\log^{18}m}}{m}
        \end{align}
        holds $m\geq Cs\mu^2\kappa{K\log^{8}m}$ with
        some sufficiently large constant $C>0$ as long as the stepsize $\eta>0$ obeys $\eta\asymp s^{-1}$ and $C_{6}>0$
        is sufficiently large. \end{lemma}
\begin{remark}
The proof of Lemma \ref{lemma:state_evo_2} and Lemma \ref{lemma:xt-xt-l_2} is inspired by the arguments used in Section H and Section I in \cite{chen2018}.
\end{remark}
\subsection{Proof for Claims (\ref{eq:exp_ratio}) and (\ref{rmse})}\label{sec:stage}
Combining the analyses in {Phase 1} and {Phase 2}, we complete
the proof for claims (\ref{eq:exp_ratio}) with $0\leq t\leq T_{\gamma}$ (\ref{eq:def-T-gamma}).
Consider the definition of $T_\gamma$ (\ref{eq:def-T-gamma}) and the incoherence between iterates and design vectors given in (\ref{eq:incohx}) and (\ref{incohh}), we arrive at 
\begin{align}
                \norm{\wt{\bm{x}}_i^{T_{\gamma}}-\bar{\bm{x}}_i}{2}&\leq \frac{\gamma}{\sqrt{2s}}\\
        \mathrm{dist}(\bm{z}^{T_{\gamma}},\bar{\bm{z}})&\leq \gamma\\
\mathrm{error}(\bm{\theta}^{T_{\gamma}},\bar{\bm{\theta}})&\leq \gamma\\
        \max_{1\leq i \leq s, 1\leq j\leq m}\left|\bm{a}_{ij}^{\mathsf{H}}\wt{\bm{x}}_i^{T_{\gamma}} \right|\cdot\|\wt{\bm{x}}_i^{T_{\gamma}}\|_2^{-1}&\lesssim\sqrt{\log m},\\
        \max_{1\leq i \leq s, 1\leq j\leq m}\left|\bm{b}_{j}^{\mathsf{H}}\wt{\bm{h}}_i^{T_{\gamma}} \right|\cdot\|\wt{\bm{h}}_i^{T_{\gamma}}\|_2^{-1}&\lesssim\frac{\mu}{\sqrt{m}}\log^2m,
\end{align}
which further implies that
\begin{align}
\max_{1\leq i \leq s, 1\leq j\leq m}\left|\bm{a}_{ij}^{\mathsf{H}}\left(\wt{\bm{x}}_i^{T_{\gamma}}-\bar{\bm{x}}_i\right)\right|\lesssim\frac{\gamma\sqrt{\log m}}{\sqrt{2s}},
\end{align}
based on the inductive hypothesis (\ref{eq:h-induction-leave_2}).
Based on these properties, we can exploit the techniques applied in \cite[Section IV]{ma2017implicit} and the triangle inequality
to prove that for $t\geq T_{\gamma}+1$,
\begin{align}
\mathrm{error}(\bm{\theta}^{t},\bar{\bm{\theta}})\leq&\text{dist}\left(\bm{x}^{t},\bar{\bm{x}}\right)\leq\text{dist}\left(\bm{z}^{t},\bar{\bm{z}}\right)\notag\\
\leq&\left(1-\frac{\eta}{16\kappa}\right)^{t-T_{\gamma}}\text{dist}\left(\bm{z}^{T_{\gamma}},\bar{\bm{z}}\right)\notag\\
\leq&\gamma\left(1-\frac{\eta}{16\kappa}\right)^{t-T_{\gamma}},\label{eq:state-2-dist}
\end{align}
where the stepsize $\eta>0$ obeys $\eta\asymp s^{-1}$ as long as $m\gg s^2\mu^2\kappa^4\max\{K,N\}\log^8m$.
It remains to prove the claim (\ref{eq:exp_ratio}) for Stage II.
Since we have already demonstrate that the ratio $\alpha_{\bm{h}_i^t}/\beta_{\bm{h}_i^t}$
increases exponentially fast in Stage I, there is
\[
\frac{\alpha_{\bm{h}_i^{T_1}}}{\beta_{\bm{h}_i^{T_1}}}\geq\frac{1}{\sqrt{2K\log K}}(1+c_{3}\eta)^{T_{1}}.
\]
By the definition of $T_{1}$ (see \eqref{eq:defn-T1}) and Lemma
\ref{lemma:iterative}, one has $\alpha_{\bm{h}_i^{T_1}}\asymp\beta_{\bm{h}_i^{T_1}}\asymp1$
and thus
\begin{equation}
\frac{\alpha_{\bm{h}_i^{T_1}}}{\beta_{\bm{h}_i^{T_1}}}\asymp1.\label{eq:SNR-T1}
\end{equation}
When it comes to $t>T_{\gamma}$, based on (\ref{eq:state-2-dist}),
we have
\begin{align*}
\frac{\alpha_{\bm{h}_i^{t}}}{\beta_{\bm{h}_i^{t}}} & \geq \frac{1- \text{dist}(\bm{h}_i^{t},\bar{\bm{h}}_i)}{\text{dist}(\bm{h}_i^{t},\bar{\bm{h}}_i)} \geq \frac{1- \text{dist}\left(\bm{z}^{t},\bar{\bm{z}}\right)}{\text{dist}\left(\bm{z}^{t},
        \bar{\bm{z}}\right)}  \\
& \geq\frac{1-\gamma/\sqrt{2}}{\gamma/\sqrt{2}}\left(1-\frac{\eta}{16\kappa}\right)^{t-T_{\gamma}}\overset{(\text{i})}{\asymp}\frac{\alpha_{\bm{h}_i^{T_1}}}{\beta_{\bm{h}_i^{T_1}}}\left(1-\frac{\eta}{16\kappa}\right)^{t-T_{\gamma}}\\
& \gtrsim\frac{1}{\sqrt{K\log K}}\left(1+c_{3}\eta\right)^{T_{1}}\left(1-\frac{\eta}{16\kappa}\right)^{t-T_{\gamma}}\\
& \overset{(\text{ii})}{ \gtrsim}\frac{1}{\sqrt{K\log K}}\left(1+c_{3}\eta\right)^{T_{\gamma}}\left(1-\frac{\eta}{16\kappa}\right)^{t-T_{\gamma}}\\
& \gtrsim\frac{1}{\sqrt{K\log K}}\left(1+c_{3}\eta\right)^{t},
\end{align*}
where (i) is derived from (\ref{eq:SNR-T1}) and the fact that $\gamma$
is a constant, (ii) arises from $T_{\gamma}-T_{1}\asymp s^{-1}$ based on Lemma \ref{lemma:iterative}, and the last inequality is satisfied as long as
$c_{3}>0$ and $\eta\asymp s^{-1}$. Likewise, we can apply the same arguments to the ratio $\alpha_{\bm{x}_i^t}/\beta_{\bm{x}_i^t}$, thereby concluding that 
\begin{align}\label{ratioab}
\frac{\alpha_{\bm{x}_i^t}}{\beta_{\bm{x}_i^t}}\gtrsim \frac{1}{\sqrt{N\log N}}\left(1+c_{4}\eta\right)^{t}.
\end{align}
Claim (\ref{rmse}) can be further derived via combining the inequality ${\rm{RMSE}}(\bm{x}_i^{t},\bar{\bm{x}}_i)  = {\beta_{\bm{x}_i^t}}/{\normn{\bm{x}_i^t}{2}}={\beta_{\bm{x}_i^t}}/{\sqrt{\alpha_{\bm{x}_i^t}^2+\beta_{\bm{x}_i^t}^2}}<\frac{ {\beta}_{\bm{x}_i^t}}{ {\alpha}_{\bm{x}_i^t}}$ and the result in (\ref{ratioab}).
\section{Conclusion}
In this paper, we proposed a  blind over-the-air computation scheme to compute the desired function of distributed sensing data without the prior knowledge of the channel information, thereby providing low-latency data aggregation in IoT networks. To harness the benefits of computational efficiency, fast convergence guarantee, regularization-free and careful {\it{initialization-free}}, the BlairComp problem was solved by randomly initialized Wirtinger flow with provable guarantees. Specifically, the statistical guarantee and fast global convergence guarantee concerning randomly initialized Wirtinger flow for solving the BlairComp problem were  provided. It demonstrated that with sufficient samples,  in the first tens iterations, the randomly initialized Wirtinger flow enables the iterates to enter a local region that enjoys strong convexity and strong smoothness, where the estimation error is sufficiently small. At the second stage of this algorithm, the estimated error experiences exponential decay.
\appendices 
\section{Preliminaries}
For $\bm{a}_{ij}\in\mathbb{C}^N$, the standard concentration inequality gives that, for $i=1,\cdots,s$,
\begin{equation}
\max_{1\leq j\leq m}\left|a_{ij,1}\right|=\max_{1\leq j\leq m}\left|\bm{a}_{ij}^{\mathsf{H}} \bar{\bm{x}}\right|\leq5\sqrt{\log m}\label{eq:max-a-i-1}
\end{equation}
with probability $1-\co\left(m^{-10}\right)$ \cite{ma2017implicit}. In addition, by applying the
standard concentration inequality, we arrive at, for $i = 1,\cdots, s$,
\begin{equation}
\max_{1\leq j\leq m}\left\Vert \bm{a}_{ij}\right\Vert _{2}\leq 3\sqrt{N}\label{eq:max-a-i-norm}
\end{equation}
with probability $1-C^\prime\exp\left(me^{-cK}\right)$ for some constants, $c,C^\prime>0$ \cite{ma2017implicit}.
\begin{lemma}\label{lemma:Hessian-UB-Stage1}
        Fix any constant $c_{0}>1$. 
        Define the population matrix $\nabla^{2}_{\bm{z}_i}F\left(\bm{z}\right)$ as
        \begin{small}
      \begin{align*}
      \left[~\begin{matrix}
      \norm{\bm{x}_i}{2}^2    \bm{I}_K &\bm{h}_i\bm{x}_i^{\mathsf{H}}-\bar{\bm{h}}_i\bar{\bm{x}}_i^{\mathsf{H}}&\bm{0}& \bar{\bm{h}}_i \bar{\bm{x}}_i^{ \top}\\
      \bm{x}_i\bm{h}_i^{\mathsf{H}}-\bar{\bm{x}}_i \bar{\bm{h}}_i^{\mathsf{H}} &  \norm{\bm{h}_i}{2}^2\bm{I}_K& \bar{\bm{x}}_i \bar{\bm{h}}_i^{ \top}&\bm{0}\\
      \bm{0}& \big(\bar{\bm{x}}_i \bar{\bm{h}}_i^{ \top}\big)^{\mathsf{H}}& \norm{\bm{x}_i}{2}^2\bm{I}_K & ({\bm{h}_i\bm{x}_i^{\mathsf{H}}-\bar{\bm{h}}_i\bar{\bm{x}}_i^{\mathsf{H}}})\mathsf{H}\\
      \big(\bar{\bm{h}}_i \bar{\bm{x}}_i^{ \top}\big)^{\mathsf{H}}& \bm{0} & ({\bm{x}_i\bm{h}_i^{\mathsf{H}}-\bar{\bm{x}}_i\bar{\bm{h}}_i^{\mathsf{H}}})\mathsf{H}& \norm{\bm{h}_i}{2}^2\bm{I}_K
      \end{matrix}~\right] 
      \end{align*}
        \end{small}    
        Suppose that $m>c_{1}s^2\mu^2K\log^{3}m$ for some sufficiently large constant
        $c_{1}>0$. Then with probability exceeding $1-\co\left(m^{-10}\right)$,
        \begin{align*}&\left\Vert \left(\bm{I}_{4K}-\eta\nabla^{2}f\left(\bm{z}\right)\right)-\left(\bm{I}_{4K}-\eta\nabla^{2}F\left(\bm{z}\right)\right) \right\Vert\\ \lesssim&\sqrt{\frac{s^2\mu^2K\log m}{m}}\max\left\{ \left\Vert \bm{z}\right\Vert _{2}^{2},1\right\} 
        \end{align*}
        \begin{align*}
        \text{and} \qquad \left\Vert \nabla^{2}f\left(\bm{z}\right)\right\Vert  & \leq 5\|\bm{z}\|_{2}^{2}+2
        \end{align*}
        hold simultaneously for all $\bm{z}$ obeying $\max_{1\leq i\leq s,1\leq l\leq m}\left|\bm{a}_{il}^{\mathsf{H}}\bm{x}_i\right|\cdot {\big\|\bm{x}_i\big\|_{2}}^{-1} \lesssim\sqrt{\log m}$ and $    \max_{1\leq i\leq s,1\leq l\leq m}\left|\bm{b}_{l}^{\mathsf{H}}\bm{h}_i\right|\cdot{\big\|\bm{h}_i\big\|_{2}}^{-1}\lesssim\frac{\mu}{\sqrt{m}}\log^2m$,
        provided that $0<\eta<\frac{c_{2}}{\max\{\left\Vert \bm{z}\right\Vert _{2}^{2},1\}}$
        for some sufficiently small constant $c_{2}>0$.\end{lemma}

\section{Proof of Lemma \ref{lemma:state_evo}}\label{sec:proof_state_evo}
 According to the Wirtinger flow gradient update rule (\ref{g2}),
and the expression $\bm{a}_{kj}^{\mathsf{H}}\bm{x}_k^{t}=x_{k\|}^{t}{a_{kj,1}^*}+\bm{a}_{kj,\perp}^{\mathsf{H}}\bm{x}_{k\perp}^{t}$
and reformulate terms, we arrive at
\begin{align}
\wt{x}_{i1}^{t+1} & =\wt{x}_{i1}^{t}+ \eta^\prime J_{i1}-\eta^\prime  J_{i2}- \eta^\prime  J_{i3},
\end{align}
where
\begin{align*}
J_{i1} &= \sum_{j=1}^m\sum_{k= 1}^s \bar{\bm{h}}_k^{\mathsf{H}}\bm{b}_j\bm{b}_j^{\mathsf{H}}\wt{\bm{h}}_i^t{a_{kj,1}}^*q_ka_{ij,1},\\
J_{i2} & = \sum_{j=1}^m\sum_{k= 1}^s\wt{\bm{h}}_k^{t\,H}\bm{b}_j\bm{b}_j^{\mathsf{H}}\wt{\bm{h}}_i^t {a_{kj,1}}^*\wt{x}_{k\|}^{t}a_{ij,1},\\
J_{i3} & = \sum_{j=1}^{m}\sum_{k=1}^s\wt{\bm{h}}_k^{t\, H}\bm{b}_j\bm{b}_j^{\mathsf{H}}\wt{\bm{h}}_i^t\bm{a}_{kj,\perp}^{\mathsf{H}}\bm{x}_{i\perp}^ta_{ij,1},\\
\eta^\prime & = \eta/\|\wt{\bm{h}}_i^t\|_2^2.
\end{align*}
We will control the above three terms $J_{i1}$, $J_{i2}$ and 
$J_{i3}$ separately in the following.
\begin{itemize}
        \item  With regard to the first term $J_{i1}$, it has
        \begin{align}
        &\sum_{j=1}^m\sum_{k= 1}^sq_k \bar{\bm{h}}_k^{\mathsf{H}}\bm{b}_j\bm{b}_j^{\mathsf{H}}\wt{\bm{h}}_i^t
        {a_{kj,1}}^* \; a_{ij,1}\notag\\
        =&
        \sum_{k=1}^sq_k
        \bar{\bm{h}}_k^{\mathsf{H}}\paren{\sum_{j=1}^m
                {a_{kj,1}}^*\, a_{ij,1}\bm{b}_j\bm{b}_j^{\mathsf{H}}}\wt{\bm{h}}_i^t.\notag
        \end{align}
        According to Lemma \ref{lemma:concentration-identity-ii} and Lemma \ref{lemma:concentration-identity-ij}, there is 
        \begin{align}
        J_{i1} = q_i\bar{\bm{h}}_i^{\mathsf{H}}\wt{\bm{h}}_i^t + r_{1},
        \end{align}
        where the size of the remaining term $r_1$ satisfies
        \begin{align}
        \abs{r_1}\lesssim \sum_{k=1}^sq_k\bar{\bm{h}}_i^{\mathsf{H}}\wt{\bm{h}}_i^t\sqrt{\frac{K}{m}\log m}\lesssim \sqrt{\frac{s^2K}{m}\log m}\cdot\bar{\bm{h}}_i^{\mathsf{H}}\wt{\bm{h}}_i^t,
        \end{align}
        based on the fact that $\|\bar{\bm{h}}_k\|^2\lesssim1$ and $\|\wt{\bm{h}}_k^t\|^2\lesssim 1$ for $k=1,\cdots,s$.
        \item Similar to the first term, the term $J_{i2}$ can be represented as
       $
        J_{i2} = \norm{\wt{\bm{h}}_i^t}{2}^2\wt{x}_{i1}^{t} + r_{2},
        $
        where the term $r_{i2}$ obeys
        \begin{align}
        \abs{r_2}\lesssim\abs{\wt{x}_{i1}^{t}}\sum_{k=1}^s\wt{\bm{h}}_k^{t\,H}\wt{\bm{h}}_i^t\sqrt{\frac{K}{m}\log m}\lesssim \sqrt{\frac{s^2K}{m}\log m}\abs{\wt{x}_{i1}^{t}}.
        \end{align}
        \item For the last term $J_{i3}$, it follows that 
 \begin{align}\label{eq:J_i3}
        &\sum_{j=1}^{m}\sum_{k=1}^s\wt{\bm{h}}_k^{t\,H}\bm{b}_j\bm{b}_j^{\mathsf{H}}\wt{\bm{h}}_i^t\bm{a}_{kj,\perp}^{\mathsf{H}}\wt{\bm{x}}_{i\perp}^ta_{ij,1}\notag\\
         =& \sum_{k=1}^s\wt{\bm{h}}_k^{t\,H}\paren{\sum_{j=1}^ma_{ij,1}\bm{a}_{kj,\perp}^{\mathsf{H}}\bm{x}_{i\perp}^t\bm{b}_j\bm{b}_j^{\mathsf{H}}}\wt{\bm{h}}_i^t.
        \end{align}
        By exploiting the random-sign sequence $\left\{ \bm{x}_i^{t,\mathrm{sgn}}\right\} $,
        one can decompose 
        \begin{align}
        &\sum_{j=1}^ma_{ij,1}\bm{a}_{kj,\perp}^{\mathsf{H}}\wt{\bm{x}}_{i\perp}^t\bm{b}_j\bm{b}_j^{\mathsf{H}}
        =\sum_{j=1}^{m}a_{ij,1}\bm{a}_{kj\perp}^{\mathsf{H}}\check{\bm{x}}_{i\perp}^{t,\mathrm{sgn}}\bm{b}_j\bm{b}_j^{\mathsf{H}}+\notag\\&\qquad\sum_{j=1}^{m}a_{ij,1}\bm{a}_{kj,\perp}^{\mathsf{H}}\left(\wt{\bm{x}}_{i\perp}^{t}-\check{\bm{x}}_{i\perp}^{t,\mathrm{sgn}}\right)\bm{b}_j\bm{b}_j^{\mathsf{H}}.\label{eq:J2}
        \end{align}
Note
  that $a_{ij,1}\bm{a}_{kj\perp}^{\mathsf{H}}\check{\bm{x}}_{i\perp}^{t,\mathrm{sgn}}\bm{b}_j\bm{b}_j^{\mathsf{H}}$ in (\ref{eq:J2})
        is statistically independent of $\xi_{ij}$ (\ref{eq:xi}) and $\bm{b}_j^{\sgn}\bm{b}_j^{\sgn\,H} = \bm{b}_j\bm{b}_j^{\mathsf{H}} $.
Hence we can consider $\sum_{j=1}^{m}a_{ij,1}\bm{a}_{kj\perp}^{\mathsf{H}}\check{\bm{x}}_{i\perp}^{t,\mathrm{sgn}}\bm{b}_j\bm{b}_j^{\mathsf{H}}$
        as a weighted sum of the $\xi_{ij}$'s and exploit the Bernstein inequality to derive that
        \begin{align}
                &\norm{\sum_{j=1}^{m}\xi_{ij}\left(a_{ij,1}\bm{a}_{kj\perp}^{\mathsf{H}}\check{\bm{x}}_{i\perp}^{t,\mathrm{sgn}}\bm{b}_j\bm{b}_j^{\mathsf{H}}\right)}{}\notag\\
                \lesssim&\sqrt{V_{1}\log m}+B_{1}\log m\label{eq:bernstein-1}
                \end{align}

        with probability exceeding $1-\co\left(m^{-10}\right)$, where 
        \[
V_{1}:=\sum_{j=1}^{m}\left|a_{ij,1}\right|^{2}\abs{\bm{a}_{kj\perp}^{\mathsf{H}}\check{\bm{x}}_{i\perp}^{t,\mathrm{sgn}}}^2\abs{\bm{b}_j\bm{b}_j^{\mathsf{H}}}^2,\]
        \[ B_{1}:=\max_{1\leq j\leq m}|a_{ij,1}|\abs{\bm{a}_{kj\perp}^{\mathsf{H}}
                \check{\bm{x}}_{i\perp}^{t,\mathrm{sgn}}}\abs{\bm{b}_j\bm{b}_j^{\mathsf{H}}}.
        \]
In view of Lemma \ref{lemma:ai-uniform-concentration} and the incoherence
        condition (\ref{eq:x_sgn_perp}) to deduce
        that with probability at least $1-\co\left(m^{-10}\right)$, 
        \begin{align*}
        V_{1} & \lesssim\norm{\sum_{j=1}^{m}\left|a_{i,1}\right|^{2}\abs{\bm{a}_{kj\perp}^{\mathsf{H}}\check{\bm{x}}_{i\perp}^{t,\mathrm{sgn}}}^2\bm{b}_j\bm{b}_j^{\mathsf{H}}}{}\norm{\bm{b}_j}{2}^2\lesssim\frac{K}{m}\left\Vert \check{\bm{x}}_{i\perp}^{t,\mathrm{sgn}}\right\Vert _{2}^{2}
        \end{align*}
        with the proviso that $m\gg \max\{K,N\}\log^{3}m$. Furthermore, the incoherence
        condition (\ref{eq:x_sgn_perp}) together
        with the fact (\ref{eq:max-a-i-1}) implies that 
        \[
        B_{1}\lesssim\frac{K}{m}\log m\left\Vert \check{\bm{x}}_{i\perp}^{t,\mathrm{sgn}}\right\Vert _{2}.
        \]
        Substitute the bounds on $V_{1}$ and \textbf{$B_{1}$ }back to (\ref{eq:bernstein-1})
        to obtain 
        \begin{equation}
        \norm{\sum_{j=1}^{m}a_{ij,1}\bm{a}_{kj\perp}^{\mathsf{H}}\check{\bm{x}}_{i\perp}^{t,\mathrm{sgn}}\bm{b}_j\bm{b}_j^{\mathsf{H}}}{}\lesssim\sqrt{\frac{K\log m}{m}}\left\Vert \check{\bm{x}}_{i\perp}^{t,\mathrm{sgn}}\right\Vert _{2}\label{eq:J2-1}
        \end{equation}
        as long as $m\gtrsim K\log^{3}m$. In addition, we move to the second
        term on the right-hand side of (\ref{eq:J2}).
        Let $\bm{u} = {\sum_{j=1}^{m}a_{ij,1}\bm{a}_{kj}^{\mathsf{H}}}\bm{z}\bm{b}_j\bm{b}_j^{\mathsf{H}}$, where $\bm{z}\in\mathbb{C}^{N-1}$ is independent with $\{\bm{a}_{kj}\}$ and $\norm{\bm{z}}{2} = 1$. Hence, we have
        \begin{align}
        &\norm{\sum_{j=1}^{m}a_{ij,1}\bm{a}_{kj,\perp}^{\mathsf{H}}\left(\wt{\bm{x}}_{i\perp}^{t}-\check{\bm{x}}_{i\perp}^{t,\mathrm{sgn}}\right)\bm{b}_j\bm{b}_j^{\mathsf{H}}}{}\notag\\
        \leq&\left\Vert \bm{u}\right\Vert _{2}\left\Vert \wt{\bm{x}}_{i\perp}^{t}-\check{\bm{x}}_{i\perp}^{t,\mathrm{sgn}}\right\Vert _{2}\lesssim\sqrt{\frac{K\log m}{m}}\left\Vert \wt{\bm{x}}_{i\perp}^{t}-\check{\bm{x}}_{i\perp}^{t,\mathrm{sgn}}\right\Vert _{2},\label{eq:J2-2}
        \end{align}
        with probability exceeding $1-\co\left(m^{-10}\right)$, as long as
        that $m\gg K\log^{3}m$. Here, the last inequality of \eqref{eq:J2-2} comes from Lemma \ref{lemma:concentration-identity-i1j}. Substituting the above two bounds (\ref{eq:J2-1})
        and (\ref{eq:J2-2}) into (\ref{eq:J2}), it yields
        \begin{align}\label{eq:sub_norm}
&       \norm{\sum_{j=1}^ma_{ij,1}\bm{a}_{kj,\perp}^{\mathsf{H}} \wt{\bm{x}}_{i\perp}^t\bm{b}_j\bm{b}_j^{\mathsf{H}}}{}\notag\\
 \lesssim       &\sqrt{\frac{K\log m}{m}}\left\Vert \check{\bm{x}}_{i\perp}^{t,\mathrm{sgn}}\right\Vert _{2}+\sqrt{\frac{K\log m}{m}}\left\Vert \wt{\bm{x}}_{i\perp}^{t}-\check{\bm{x}}_{i\perp}^{t,\mathrm{sgn}}\right\Vert _{2}.
        \end{align}
        Combining (\ref{eq:J_i3}) and (\ref{eq:sub_norm}), we arrive at
        \begin{align}
&       \abs{J_{i3}}\notag\\
        {\lesssim}&\sqrt{\frac{s^2K\log m}{m}}\left\Vert \wt{\bm{x}}_{i\perp}^{t}\right\Vert _{2}+\sqrt{\frac{s^2K\log m}{m}}\left\Vert \wt{\bm{x}}_{i\perp}^{t}-\check{\bm{x}}_{i\perp}^{t,\mathrm{sgn}}\right\Vert _{2},
        \end{align}
by exploiting the fact that $\|\wt{\bm{h}}_k^t\|^2\lesssim 1$ for $k=1,\cdots,s$ and
        the triangle inequality $\left\Vert \check{\bm{x}}_{i\perp}^{t,\mathrm{sgn}}\right\Vert _{2}\leq\left\Vert \wt{\bm{x}}_{i\perp}^{t}\right\Vert _{2}+\left\Vert \wt{\bm{x}}_{i\perp}^{t}-\check{\bm{x}}_{i\perp}^{t,\mathrm{sgn}}\right\Vert _{2}$.
        \item  Collecting the bounds for $J_{i1}$, $J_{i2}$ and $J_{i3}$, we arrive at 
        \begin{align}\label{eq:x-t+1-signal}
        \wt{x}_{i1}^{t+1}  &=\wt{x}_{i1}^{t}+ \eta^\prime J_{i1}-\eta^\prime  J_{i2}- \eta^\prime  J_{i3}\notag\\
        & = \wt{x}_{i1}^{t}+\eta q_i\bar{\bm{h}}_i^{\mathsf{H}}\bm{h}_i^t/\|\wt{\bm{h}}_i^t\|_2^2 - \eta \wt{x}_{i1}^{t}+R\notag\\
        & = \left(1-\eta\right) {x}_{i1}^{t}+\eta q_i{\bar{\bm{h}}}_i^{\mathsf{H}}\bm{h}_i^t/\|\wt{\bm{h}}_i^t\|_2^2+R,
        \end{align}
        where the residual term $R$ follows that
        \begin{align}
        \abs{R}\lesssim &\frac{\eta}{\|\wt{\bm{h}}_i^t\|_2^2} \sqrt{\frac{s^2K}{m}\log m}\Big(bar{\bm{h}}_i^{\mathsf{H}}\bm{h}_i^t+\abs{\wt{x}_{i1}^{t}}+\left\Vert \wt{\bm{x}}_{i\perp}^{t}\right\Vert _{2}\notag\\
        &+\left\Vert \wt{\bm{x}}_{i\perp}^{t}-\check{\bm{x}}_{i\perp}^{t,\mathrm{sgn}}\right\Vert _{2}\Big).
        \end{align}
        Substituting the hypotheses (\ref{subeq:induction}) into (\ref{eq:x-t+1-signal})
        and in view of the fact $\alpha_{\bm{x}_i^{t}}=\langle\bm{x}^{t},\bar{\bm{x}}\rangle/\normn{\bar{\bm{x}}_i}{2}$ and the assumption that $\normn{\bar{\bm{h}}_i}{2} = \normn{\bar{\bm{x}}_i}{2} = q_i$ for $i = 1,\cdots, s$, one has
        \begin{small}
        \begin{align}
                &\alpha_{\bm{x}_i^{t+1}} \notag\\& =\left(1-\eta\right) \alpha_{\bm{x}_i^{t}}+\eta^{\prime\prime }q_i\bar{\bm{h}}_i^{\mathsf{H}} \wt{\bm{h}}_i^t+\co\left(\eta^{\prime\prime } \sqrt{\frac{s^2K}{m}\log m}\alpha_{\bm{x}_i^{t}}\right)\notag\\
                &+\co\left(\eta^{\prime\prime }  \sqrt{\frac{s^2K}{m}\log m}\beta_{\bm{x}_i^{t}}\right) +\co\paren{\eta^{\prime\prime } \sqrt{\frac{s^2K}{m}\log m}\cdot{\alpha}_{{\bm{h}}_i^{t}}}\notag\\
                &+\co\left(\eta^{\prime\prime } \alpha_{\bm{x}_i^{t}}\left(1+\frac{1}{s\log m}\right)^{t}C_{3}\sqrt{\frac{s\mu^2N\log^{8}m}{m}}\right)\nonumber \\
                & =(1-\eta+\frac{\eta q_i \psi_{{\bm{x}}_i^{t}}}{{\alpha}_{{\bm{x}}_i^{t}}^2+{\beta}_{\bm{x}_i^{t}}^2}) \alpha_{\bm{x}_i^{t}}+\eta(1-\rho_{\bm{x}_i^{t}})\frac{{q_i\alpha}_{{\bm{h}}_i^{t}}}{{\alpha}_{{\bm{h}}_i^{t}}^2+{\beta}_{{\bm{h}}_i^{t}}^2},\label{eq:alpha-t-iterative}
                \end{align}
        \end{small}
        where $\eta^{\prime\prime } = \eta/(q_i\normn{{\bm{h}}_i^{t}}{2}^2)$, for some $|\psi_{\bm{x}_i^{t}}|,|\rho_{\bm{x}_i^{t}}|\ll\frac{1}{\log m}$, provided that \begin{subequations}
                \begin{align}
        &       \sqrt{\frac{s^2K\log m}{q_i^2m}}  \ll\frac{q_i}{\log m},\label{eq:alpha-t-as-long-as-1}\\
        &       \sqrt{\frac{s^2K\log m}{q_i^2m}}\beta_{\bm{x}_i^{t}}  \ll\frac{q_i}{\log m}\alpha_{\bm{x}_i^{t}},\label{eq:alpha-t-as-long-as-2}\\
        &       \left(1+\frac{1}{s\log m}\right)^{t}C_{3}\sqrt{\frac{s\mu^2N\log^{8}m}{q_i^2m}} \ll\frac{q_i}{\log m},\label{eq:alpha-t-as-long-as-3}
                \end{align}
                where the parameter $q_i$ is assumed to be $0<q_i\leq 1$.
        \end{subequations}Therein, the first condition (\ref{eq:alpha-t-as-long-as-1})
        naturally holds as long as $m\gg s^2K\log^{3}m$. In addition,
        the second condition (\ref{eq:alpha-t-as-long-as-2}) holds true since
        $\beta_{\bm{x}_i^{t}}\leq\|\bm{x}_i^{t}\|_{2}\lesssim\alpha_{\bm{x}_i^{t}}\sqrt{\log^5 m}$
         (based on (\ref{eq:x-induction-norm-relative}))
        and $m\gg s^2K\log^{8}m$. For the last condition (\ref{eq:alpha-t-as-long-as-3}),
        we have for $t\leq T_{1}=\co\left(s\log \max\{K,N\}\right)$, 
        \[
        \left(1+\frac{1}{s\log m}\right)^{t}=\co\left(1\right),
        \]
        which further implies 
\begin{align*}
&\left(1+\frac{1}{s\log m}\right)^{t}C_{3}\sqrt{\frac{s\mu^2N\log^{8}m}{q_i^2m}}\notag\\
\lesssim &C_{3}\sqrt{\frac{s\mu^2N\log^{8}m}{q_i^2m}}\ll\frac{q_i}{\log m}
\end{align*}

        as long as the number of samples obeys $m\gg s\mu^2N\log^{10}m$. This concludes
        the proof. 
\end{itemize}

        Despite it turns to be more tedious when proving (\ref{eq:evo_alpha_H}), similar arguments used above can be applied to the proof of (\ref{eq:evo_alpha_H}).
        Specifically, according to the Wirtinger flow gradient update rule (\ref{g1}),
        the signal component $
        \langle\bar{\bm{h}}_i,\widetilde{\bm{h}}^t_i\rangle$ can be represented as follows 
        \begin{align*}
        &\bar{\bm{h}}_i^{\mathsf{H}}\wt{\bm{h}}_{i}^{t+1}\\
         = & \bar{\bm{h}}_i^{\mathsf{H}}\wt{\bm{h}}_{i}^{t} -\frac{\eta}{\|\wt{\bm{x}}_i^t\|_2^2}\sum_{j=1}^m\bigg(\sum_{k=1}^{ s}\bm{b}_j^{\mathsf{H}}\wt{\bm{h}}_k^t\wt{\bm{x}}_k^{t\,H}\bm{a}_{kj}-{y}_j\bigg)\bar{\bm{h}}_i^{\mathsf{H}}\bm{b}_j\bm{a}_{ij}^{\mathsf{H}}\wt{\bm{x}}_i^t.
        \end{align*}
        Expanding this expression using $\bm{a}_{kj}^{\mathsf{H}}\bm{x}_k^{t}=x_{k\|}^{t}{a_{kj,1}^*}+\bm{a}_{kj,\perp}^{\mathsf{H}}\bm{x}_{k\perp}^{t}$
        and rearranging terms, we are left with 
        \begin{align}
        \bar{\bm{h}}_i^{\mathsf{H}}\wt{\bm{h}}_{i}^{t+1} & =\bar{\bm{h}}_i^{\mathsf{H}}\wt{\bm{h}}_{i}^{t}-\eta_i^\prime L_{i1}+ \eta^\prime_i L_{i2}+ \eta_i^\prime L_{i3},
        \end{align}
        where
        \begin{align*}
        L_{i1}&= \sum_{j=1}^m\sum_{k= 1}^s \bar{\bm{h}}_i^{\mathsf{H}}\bm{b}_j\bm{b}_j^{\mathsf{H}}\wt{\bm{h}}_k^t\wt{\bm{x}}_k^{t\,H}\bm{a}_{kj}\bm{a}_{ij}^{\mathsf{H}}\bm{x}_i,\\
        L_{i2} & = \sum_{j=1}^m\sum_{k= 1}^s\bar{\bm{h}}_i^{\mathsf{H}}\bm{b}_j \bm{b}_j^{\mathsf{H}}\bar{\bm{h}}_k {a_{kj,1}}q_k
        {a_{ij,1}}*t\, \wt{x}_{i1}^{t},\\
        L_{i3} & = \sum_{j=1}^m\sum_{k= 1}^s\bar{\bm{h}}_i^{\mathsf{H}}\bm{b}_j \bm{b}_j^{\mathsf{H}}\bar{\bm{h}}_k\bm{a}_{ij,\perp}^{\mathsf{H}}\bm{x}_{i\perp}^ta_{kj,1}q_k,\\
        \eta_i^\prime &= \eta/\normn{\wt{\bm{x}}_i^t}{2}^2.
        \end{align*}
Here, $ L_{i1}$, $L_{i2}$ and $L_{i3}$ can be controlled via  the strategies  exploited to control $J_{i1}$, $J_{i2}$ and $J_{i3}$. The proof of (\ref{eq:evo_beta_X}) is based on similar arguments as above.

\section{Proof of (\ref{eq:dist_l_x}) in Lemma \ref{lemma:xt-xt-l-signal}}\label{sec:proof_leave_one}
By applying the arguments in \cite[Appendix F]{dong2018}, it yields that
 \begin{align}\label{78}
 &\mathrm{dist}\left(\bm{x}_i^{t+1,(l)},\widetilde{\bm{x}}_i^{t+1}\right)\notag\\
 \leq &\kappa\sqrt{\sum_{k = 1}^s \max\left\{\left|\frac{\omega_i^{t+1}}{\omega_i^t}\right|,\left|\frac{\omega_i^{t}}{\omega_i^{t+1}}\right|\right\}^2\|\bm{J}_k\|^2},
 \end{align}
 where $\omega_i^t$ is the alignment parameter and
\begin{align}
\bm{J}_k =  {{\omega_k^t}}\bm{x}_k^{t+1}- {{\omega_{k,\text{mutual}}^{t,(l)}}}\bm{x}_k^{t+1,(l)},
\end{align}
where $\omega_{k,\text{mutual}}^{t,(l)}$ is defined in (\ref{eq:mutual_seq}).
 According to (\ref{eq:def-xl}) and (\ref{eq:mutual_seq}), we arrive at
\begin{align}
& {\omega_i^t}{x}_{i1}^{t+1}-\omega_{i,\text{mutual}}^{t,(l)}{x}_{i1}^{t+1,(l)}\nonumber\\
=&\widetilde{{x}}_{i1}^{t}-\widehat{{x}}_{i1}^{t,(l)}-\eta^\prime\bm{e}_{1}^{\top} \paren{\nabla_{\bm{x}_i} f\left(\widetilde{\bm{z}}^{t}\right) - \nabla_{\bm{x}_i} f^{(l)}\big(\widehat{\bm{z}}_i^{t,(l)}\big)} \notag\\&-\eta^\prime \bigg(\sum_{k=1}^{ s}\widehat{\bm{h}}_i^{t,(l)H}\bm{b}_l\bm{a}_{kl}^{\mathsf{H}}\widehat{\bm{x}}_i^{t,(l)}-\bar{\bm{h}}_k^{\mathsf{H}}\bm{b}_l\bm{a}_{kl}^{\mathsf{H}}\bar{\bm{x}}_k\bigg)
\bm{b}_l^{\mathsf{H}}\widehat{\bm{h}}_i^{t,(l)}{a}_{il,1}\nonumber \label{eq:xt-xt-l-signal},
\end{align}
where the stepsize $\eta^\prime = \eta/\normn{\wt{\bm{h}}_i^t}{2}^2$.
It follows from the fundamental theorem of calculus \cite[Theorem 4.2]{pothoven2013real} that
\begin{align}
&\widetilde{{x}}_{i1}^{t+1}-\widehat{{x}}_{i1}^{t+1,(l)}\notag\\ =&\Brac{\widetilde{{x}}_{i1}^{t}-\widehat{{x}}_{i1}^{t,(l)}-\eta^\prime\paren{ \int_{0}^{1}\bm{e}_{1}^{\top}\nabla^{2}_{\bm{x}_i}f\left(\bm{z}\left(\tau\right)\right)\mathrm{d}\tau}\left[\begin{array}{c}
\widetilde{\bm{x}}_{i}^{t}-\widehat{\bm{x}}_{i}^{t,(l)} \\
\overline{\widetilde{\bm{x}}_{i}^{t}-\widehat{\bm{x}}_{i}^{t,(l)}}
\end{array}\right]} \notag\\
-&
\eta^\prime\brac{\bigg(\sum_{k=1}^{ s}\widehat{\bm{h}}_i^{t,(l)H}\bm{b}_l\bm{a}_{kl}^{\mathsf{H}}\widehat{\bm{x}}_i^{t,(l)}-
        \bar{\bm{h}}_k^{\mathsf{H}}\bm{b}_l\bm{a}_{kl}^{\mathsf{H}}\bar{\bm{x}}_k\bigg)\bm{b}_l^{\mathsf{H}}\widehat{\bm{h}}_i^{t,(l)}{a}_{il,1}},
\end{align}
where $\bm{z}\left(\tau\right)=\widetilde{\bm{z}}^{t}+\tau\left(\widehat{\bm{z}}^{t,\left(l\right)}-\widetilde{\bm{z}}^{t}\right)$ with $0\leq \tau\leq 1$ and the Wirtinger Hessian with respect to $\bm{x}_i$ is
\begin{align}\label{hessian}
\nabla^{2}_{\bm{x}_i}f\left(\bm{z}\right) = \left[~\begin{matrix}
\bm{D}&\bm{E}\\
\bm{E}^{\mathsf{H}}& {(\bm{D}^\mathsf{H})^\top}
\end{matrix}~\right],
\end{align}
with
\[
 \bm{D} = \sum_{j=1}^{m}|\bm{b}^{\mathsf{H}}_{j}\bm{h}_i|^2\bm{a}_{ij}\bm{a}_{ij}^{\mathsf{H}}, 
 \;
\bm{E} =\sum_{j=1}^m  \bm{b}_{j}\bm{b}_{j}^{\mathsf{H}}\bm{h}_i(\bm{a}_{ij}\bm{a}_{ij}^{\mathsf{H}}\bm{x}_i)^\top. 
\] 
\begin{itemize}
\item 
We begin by controlling the second term of (\ref{eq:xt-xt-l-signal}).
Based on (\ref{eq:wth}) and the hypothesis (\ref{eq:h-induction-leave}), we obtain 
\[      \max_{1\leq i \leq s, 1\leq l\leq m}\left|\bm{b}_{l}^{\mathsf{H}}\widehat{\bm{h}}_i^{t,(l)} \right|\cdot\|\widehat{\bm{h}}_i^{t,(l)}\|_2^{-1}\lesssim\frac{\mu}{\sqrt{m}}\log^2m.\]
Along with the standard concentration results 
$
\left|\bm{a}_{il}^{\mathsf{H}}\bm{x}_i^{t,(l)}\right|\lesssim\sqrt{\log m}\big\|\bm{x}_i^{t,(l)}\big\|_{2},
$ one has
\begin{align}\label{eq:hbax}
&\abs{\bigg(\sum_{k=1}^{ s}\widehat{\bm{h}}_i^{t,(l)H}\bm{b}_l\bm{a}_{kl}^{\mathsf{H}}\widehat{\bm{x}}_i^{t,(l)}-\bar{\bm{h}}_k^{\mathsf{H}}\bm{b}_l\bm{a}_{kl}^{\mathsf{H}}\bar{\bm{x}}_k\bigg)\bm{b}_l^{\mathsf{H}}\widehat{\bm{h}}_i^{t,(l)}{a}_{il,1}}\notag\\
\lesssim&\frac{s\mu^2 {\log^5 m}}{m}\left\Vert \widehat{\bm{x}}_i^{t,(l)}\right\Vert _{2}.
\end{align}
\item It remains to bound the first term in (\ref{eq:xt-xt-l-signal}). To achieve this, we first utilize the decomposition
$
\bm{a}_{ij}^{\mathsf{H}}\paren{{\widetilde{\bm{x}}_{i}^{t}-\widehat{\bm{x}}_{i}^{t,(l)}}}={a_{ij,1}}^*
\paren{{\widetilde{{x}}_{i1}^{t}-\widehat{{x}}_{i1}^{t,(l)}}}+
\bm{a}_{ij,\perp}^{\mathsf{H}}\paren{{\widetilde{\bm{x}}_{i\perp}^{t}-\widehat{\bm{x}}_{i\perp}^{t,(l)}}}
$
to obtain that
\begin{small}
\begin{align*}
        &\!\!\!\!\!\!\bm{e}_{1}^{\top}\paren{\nabla^{2}_{\bm{x}_i}f\left(\bm{z}\left(\tau\right)\right)\mathrm{d}\tau}\left[\begin{array}{c}
        \widetilde{\bm{x}}_{i}^{t}-\widehat{\bm{x}}_{i}^{t,(l)} \\
        \overline{\widetilde{\bm{x}}_{i}^{t}-\widehat{\bm{x}}_{i}^{t,(l)}}
        \end{array}\right] = \omega_{1}\left(\tau\right)+\omega_{2}\left(\tau\right)+\omega_{3}\left(\tau\right),\end{align*}
\end{small}
where
\begin{align*}
\omega_{1}\left(\tau\right)& = \sum_{j=1}^{m}|\bm{b}^{\mathsf{H}}_{j}\bm{h}_i(\tau)|^2{a}_{ij,1}{a_{ij,1}^*}\paren{{\widetilde{{x}}_{i1}^{t}-\widehat{{x}}_{i1}^{t,(l)}}},\\
\omega_{2}\left(\tau\right)& = \sum_{j=1}^{m}|\bm{b}^{\mathsf{H}}_{j}\bm{h}_i(\tau)|^2{a}_{ij,1}\bm{a}_{ij,\perp}^{\mathsf{H}}\paren{{\widetilde{\bm{x}}_{i\perp}^{t}-
                \widehat{\bm{x}}_{i\perp}^{t,(l)}}},\\
\omega_{3}\left(\tau\right)& = \sum_{j=1}^m 
 \bm{b}_{j}^{\mathsf{H}}\bm{h}_i(\tau)\bm{a}_{ij}^{\mathsf{H}}\bm{x}_i(\tau){b}_{j,1}\bm{a}_{ij}^{\top}\paren{{\widetilde{\bm{x}}_{i}^{t}-\widehat{\bm{x}}_{i}^{t,(l)}}}.
\end{align*}
Based on Lemma \ref{lemma:Hessian-UB-Stage1}, Lemma \ref{lemma:concentration-identity-bha} and the fact $\norm{\bm{b}_j}{2} = \sqrt{K/m}$, by exploiting the techniques in Appendix \ref{sec:proof_state_evo}, 
$\omega_{1}\left(\tau\right)$, $\omega_{2}\left(\tau\right)$ and $\omega_{3}\left(\tau\right)$ can be bounded as follows:
\begin{align}
        \omega_{1}\left(\tau\right)&=\norm{\bm{h}_i(\tau)}{2}^2\big(\widetilde{{x}}_{i1}^{t}-\widehat{{x}}_{i1}^{t,(l)}\big)\notag\\&\quad+\co\left(\sqrt{\frac{s^2\mu^2K\log m}{m}}\bigg(\widetilde{{x}}_{i1}^{t}-\widehat{{x}}_{i1}^{t,(l)}\bigg)\right)\label{eq:bound_w_1}\\      
        \abs{\omega_{2}(\tau) }&\lesssim\sqrt{\frac{K{\log^2 m}}{m}}\Big(\left\Vert \widetilde{\bm{x}}_{i\perp}^{t}-\widehat{\bm{x}}_{i\perp}^{t,\left(l\right)}\right\Vert _{2}\notag\\
        &+\left\Vert \widetilde{\bm{x}}_{i\perp}^{t}-\widehat{\bm{x}}_{i\perp}^{t,\left(l\right)}-\widetilde{\bm{x}}_{i\perp}^{t,\text{sgn}}-\widehat{\bm{x}}_{i\perp}^{t,\text{sgn},\left(l\right)}\right\Vert _{2}\Big)\label{eq:bound_w_2}\\
\omega_{3}\left(\tau\right)&=\abs{{h}_{i1}(\tau)}\paren{{\widetilde{\bm{x}}_{i}^{t}-\widehat{\bm{x}}_{i}^{t,(l)}}}^{\mathsf{H}}\bm{x}_i(\tau)\notag\\
        &+\co\left(\frac{1}{\log^5 m}\norm{{\widetilde{\bm{x}}_{i}^{t}-\widehat{\bm{x}}_{i}^{t,(l)}}}{2}\right)\label{eq:bound_w_3}
\end{align}
with probability at least $1-\co(m^{-10})$, provided that $m\gg \mu^2K\log^{13}m$.
\item
Combining the bounds (\ref{eq:hbax}) (\ref{eq:bound_w_1}), (\ref{eq:bound_w_2}) and (\ref{eq:bound_w_3}), one has
\begin{align*}
&\widetilde{{x}}_{i1}^{t+1}-\widehat{{x}}_{i1}^{t+1,(l)}\\ =&\paren{1-\eta \frac{\int_{0}^{1}\normn{\bm{h}_i(\tau)}{2}^2\mathrm{d}\tau}{\normn{\wt{\bm{h}}_i^t}{2}^2} +\co\paren{\eta^\prime\sqrt{ \frac{s^2\mu^2K\log m}{m}}}}\cdot\notag\\&\paren{\widetilde{{x}}_{i1}^{t}-\widehat{{x}}_{i1}^{t,(l)}}
+
\co\paren{\eta^\prime\frac{s\mu^2 {\log^5 m}}{m}\left\Vert \widehat{\bm{x}}_i^{t,(l)}\right\Vert_{2}} \notag\\
& +\mathcal{O}\Bigg(\eta^\prime\sqrt{\frac{K{\log^2 m}}{m}}\bigg(\left\Vert \widetilde{\bm{x}}_{i\perp}^{t}-\widehat{\bm{x}}_{i\perp}^{t,\left(l\right)}\right\Vert _{2}\notag\\
&+\left\Vert \widetilde{\bm{x}}_{i\perp}^{t}-\widehat{\bm{x}}_{i\perp}^{t,\left(l\right)}-\widetilde{\bm{x}}_{i\perp}^{t,\text{sgn}}-\widehat{\bm{x}}_{i\perp}^{t,\text{sgn},\left(l\right)}\right\Vert _{2}\bigg)\Bigg)\notag\\
&+\co\left(\eta^\prime\frac{1}{\log^5 m}\norm{{\widetilde{\bm{x}}_{i}^{t}-\widehat{\bm{x}}_{i}^{t,(l)}}}{2}\right)\notag\\
&+ \eta^\prime\sup_{0\leq \tau\leq 1}\abs{{h}_{i1}(\tau)}\paren{{\widetilde{\bm{x}}_{i}^{t}-\widehat{\bm{x}}_{i}^{t,(l)}}}^{\mathsf{H}}\bm{x}_i(\tau).
\end{align*}
By exploiting similar arguments in Appendix E in \cite{chen2018}, we can arrive at 
\begin{align*}
&\dist\paren{{{x}}_{i1}^{t+1,(l)},\widetilde{{x}}_{i1}^{t+1}} =\abs{\widetilde{{x}}_{i1}^{t+1}-\widehat{{x}}_{i1}^{t+1,(l)}}\cdot\normn{\bar{\bm{x}}_i}{2}^{-1} \notag\\
\leq&\kappa\abs{\widetilde{{x}}_{i1}^{t+1}-\widehat{{x}}_{i1}^{t+1,(l)}} \\
{\leq}&\paren{1-\eta+\varrho_2} \alpha_{{\bm{x}}_i^{t}}\left(1+\frac{1}{s\log m}\right)^{t}C_{2}\frac{s\mu^2\kappa\sqrt{N\log^{13}m}}{m}\\
{\leq}&\alpha_{{\bm{x}}_i^{t+1}}\left(1+\frac{1}{s\log m}\right)^{t+1}C_{2}\frac{s\mu^2\kappa\sqrt{N\log^{13}m}}{m}
\end{align*}
for some $|\varrho_{2}|\ll\frac{1}{\log m}$ provided that $m\geq Cs\mu^2\kappa N\log^{12} m$ for some sufficiently large constant $C>0$. 
\end{itemize}
\section{Technical Lemmas}
\begin{lemma}
        \label{lemma:concentration-identity-ii}
        Suppose $m\gg K\log^{3}m$.
        With probability exceeding $1-\co\left(m^{-10}\right)$, we have 
      $
        \left\Vert \sum_{j=1}^{m}{a_{ij,1}^*}a_{ij,1}\bm{b}_{j}\bm{b}_{j}^{\mathsf{H}}-\bm{I}_{K}\right\Vert \lesssim\sqrt{\frac{K}{m}\log m}.
     $
        
\end{lemma}
\begin{lemma}\label{lemma:concentration-identity-ij}
        Suppose $m\gg K\log^{3}m$.
        For $k\neq i$, we have 
        $
        \left\Vert \sum_{j=1}^{m}{a_{kj,1}^*}a_{ij,1}\bm{b}_{j}\bm{b}_{j}^{\mathsf{H}}\right\Vert \lesssim\sqrt{\frac{K}{m}\log m},
      $
        $
        \left\Vert \sum_{j=1}^{m}\abs{a_{kj,1}}\abs{a_{ij,1}}\bm{b}_{j}\bm{b}_{j}^{\mathsf{H}}\right\Vert \lesssim\sqrt{\frac{K}{m}\log m},
        $ with probability exceeding $1-\co\left(m^{-10}\right)$.
        
\end{lemma}
\begin{lemma}\label{lemma:concentration-identity-i1j}
        Suppose $m\gg K\log^{3}m$ and $\bm{z}\in\mathbb{C}^{N-1}$ with $\norm{\bm{z}}{2} = 1$ is independent with $\{\bm{a}_{kj}\}$ .
        With probability exceeding $1-\co\left(m^{-10}\right)$, we have 
        $
        \left\Vert \sum_{j=1}^{m}a_{ij,1}\bm{a}_{kj,\perp}^{\mathsf{H}}\bm{z}\bm{b}_j\bm{b}_{j}^{\mathsf{H}}\right\Vert \lesssim\sqrt{\frac{K}{m}\log m}.
        $
        \begin{remark}
Lemma \ref{lemma:concentration-identity-ij}, Lemma \ref{lemma:concentration-identity-i1j} and Lemma \ref{lemma:concentration-identity-ii} can be proven by applying the arguments in \cite[Section D.3.3]{ma2017implicit}.
        \end{remark}
\end{lemma}
\begin{lemma}\label{lemma:concentration-identity-bha}
        Suppose $m\gg (\mu^2/\delta^2)N\log^{5}m$.
        With probability exceeding $1-\co\left(m^{-10}\right)$, we have 
       $
        \norm{\sum_{j=1}^{m}\abs{\bm{b}^{\mathsf{H}}_{j}\bm{h}_i}^2\bm{a}_{ij,\perp}\bm{a}_{ij,\perp}^{\mathsf{H}} -\norm{\bm{h}_i}{2}^2\bm{I}_{N-1} }{} \lesssim\delta\norm{\bm{h}_i}{2}^2,
      $
        obeying $\max_{ 1\leq l\leq m}\left|\bm{b}_{l}^{\mathsf{H}}\bm{h}_i\right|\cdot{\big\|\bm{h}_i\big\|_{2}}^{-1} \lesssim\frac{\mu}{\sqrt{m}}\log^2m$. Furthermore, there is 
        $\norm{\sum_{j=1}^{m}\sum_{k=1}^{ s}{b}_{j,1}\bm{b}_j^{\mathsf{H}}\bm{h}_i\bm{a}_{ij}\bm{a}_{kj}^{\mathsf{H}}-{{h}_{i1}}\bm{I}_N}{}\lesssim\delta\norm{\bm{h}_i}{2},$  with probability exceeding $1-\co\left(m^{-10}\right)$, provided  $m\gg (\mu/\delta^2)s^2N\log^3 m$.     
\end{lemma}
\begin{proof}
Please refer to Lemma 11 and Lemma 12 in \cite{dong2018}.
\end{proof}
\begin{lemma}\label{lemma:concentration-U1_full}
        Suppose the sampling size $m\gg s \mu^2\sqrt{N\log^9 m}$, then
        with probability exceeding $1-\co\left(m^{-10}\right)$, we have 
       $\norm{\sum_{j=1}^{m}\sum_{k=1}^{ s}\bm{h}_k^{\mathsf{H}}\bm{b}_j\bm{b}_j^{\mathsf{H}}\bm{h}_i\bm{a}_{ij}\bm{a}_{kj}^{\mathsf{H}}-\norm{\bm{h}_i}{2}^2\bm{I}_N}{}\lesssim\frac{s \mu^2\sqrt{K\log^9 m}}{m}\norm{\bm{h}_i}{2}^2,
        $
        obeying $\max_{1\leq i\leq s,1\leq j\leq m}\left|\bm{b}_{j}^{\mathsf{H}}\bm{h}_i\right|\cdot{\big\|\bm{h}_i\big\|_{2}}^{-1} \lesssim\frac{\mu}{\sqrt{m}}\log^2m$. 
\end{lemma}
\begin{lemma}\label{lemma:concentration-U3_full}
        Suppose the sampling size follows that $m\gg s \mu^2\sqrt{N\log^5 m}$.
        With probability exceeding $1-\co\left(m^{-10}\right)$, we have 
$
        \norm{\sum_{j=1}^{m}\sum_{k=1}^{ s}\bar{\bm{h}}_k^{\mathsf{H}}\bm{b}_j\bm{b}_j^{\mathsf{H}}\bm{h}_i\bm{a}_{ij}\bm{a}_{kj}^{\mathsf{H}}-({\bar{\bm{h}}_i^{\mathsf{H}}\bm{h}_i})\bm{I}_N}{}
    \lesssim\frac{s \mu^2\sqrt{K\log^5 m}}{m}\abs{\bar{\bm{h}}_i^{\mathsf{H}}\bm{h}_i},
$
        obeying  $\max_{ 1\leq l\leq m}\left|\bm{b}_{l}^{\mathsf{H}}\bar{\bm{h}}_i\right|\cdot{\big\|\bar{\bm{h}}_i\big\|_{2}}^{-1} \leq\frac{\mu}{\sqrt{m}}$ and $\max_{ 1\leq l\leq m}\left|\bm{b}_{l}^{\mathsf{H}}\bm{h}_i\right|\cdot{\big\|\bm{h}_i\big\|_{2}}^{-1} \lesssim\frac{\mu}{\sqrt{m}}\log^2m$.
        
\end{lemma}
\begin{remark}
The proof of Lemma \ref{lemma:concentration-U1_full} and \ref{lemma:concentration-U3_full} exploits the same strategy as \cite[Section K]{chen2018} does.
\end{remark}
\begin{lemma}\label{lemma:ai-uniform-concentration}Suppose that $\bm{a}_{ij}$ and $\bm{b}_j$ follows the definition in Section \ref{form}. $1\leq i\leq s, 1\leq j\leq m$. Consider any
        $\epsilon>3/n$ where $n = \max\{K,N\}$. Let 
$
\mathcal{S}:=\left\{ \bm{z}\in\mathbb{C}^{N-1}\Big| \max_{1\leq j\leq m}\left|\bm{a}_{ij,\perp}^{\mathsf{H}}\bm{z}\right|\leq\beta\left\Vert \bm{z}\right\Vert _{2}\right\} ,
$
        where $\beta$ is any value obeying $\beta\geq c_{1}\sqrt{\log m}$
        for some sufficiently large constant $c_{1}>0$. Then with probability
        exceeding $1-\co\left(m^{-10}\right)$, one has 
        \begin{enumerate}
                \item $\left|\sum_{j=1}^{m}\left|a_{ij,1}\right|^{2}\absn{\bm{a}_{kj\perp}^{\mathsf{H}}\bm{z}}^2\bm{b}_{j}\bm{b}_{j}^{\mathsf{H}} - \norm{\bm{z}}{2} \bm{I}_K\right|\leq\epsilon\left\Vert \bm{z}\right\Vert _{2}$
                for all $\bm{z}\in\mathcal{S}$, provided that $m\geq c_{0}\max\left\{ \frac{1}{\epsilon^{2}}n\log n,\text{ }\frac{1}{\epsilon}\beta^2 n\log^{2}m\right\} $.
                \item $\left|\sum_{j=1}^{m}\left|a_{ij,1}\right|\absn{\bm{a}_{kj\perp}^{\mathsf{H}}\bm{z}}\bm{b}_{j}\bm{b}_{j}^{\mathsf{H}}\right|\leq\epsilon\left\Vert \bm{z}\right\Vert _{2}$
                for all $\bm{z}\in\mathcal{S}$, provided that $m\geq c_{0}\max\left\{ \frac{1}{\epsilon^{2}}n\log n,\text{ }\frac{1}{\epsilon}\beta n\log^{\frac{1}{2}}m\right\} $. 
        \end{enumerate}
        Here, $c_{0}>0$ is some sufficiently large constant. \end{lemma}
\begin{proof}
Please refer to Lemma 12 in \cite{chen2018}.
\end{proof}
\bibliography{Reference} 
\bibliographystyle{ieeetr}
\end{document}